\documentclass[sigconf]{acmart}

\usepackage{amsmath}
\usepackage{amssymb}  
\usepackage{amsthm}
\usepackage{algpseudocode}
\usepackage{algorithm}
\usepackage{array}
\usepackage{tcolorbox}
\usepackage{url}
\usepackage{graphics} 
\usepackage{enumerate}
\usepackage{booktabs}
\usepackage[bottom]{footmisc}
\usepackage{setspace}
\let\Algorithm\algorithm
\renewcommand\algorithm[1][]{\Algorithm[#1]\setstretch{1.1}}

\usepackage{hyperref}
\hypersetup{
	colorlinks=true, 
    linkcolor=blue,  
    citecolor=blue,  
    filecolor=blue,  
    urlcolor=blue    
}

\newtheorem{thr}{Theorem}
\newtheorem{prop}[thr]{Proposition}
\newtheorem{remk}[thr]{Remark}
\newtheorem{defn}[thr]{Definition}

\copyrightyear{2019}
\acmYear{2019}
\setcopyright{acmcopyright}
\acmConference[Mobihoc '19]{The Twentieth ACM International Symposium on Mobile Ad Hoc Networking and Computing}{July 2--5, 2019}{Catania, Italy}
\acmBooktitle{The Twentieth ACM International Symposium on Mobile Ad Hoc Networking and Computing (Mobihoc '19), July 2--5, 2019, Catania, Italy}
\acmPrice{15.00}
\acmDOI{10.1145/3323679.3326520}
\acmISBN{978-1-4503-6764-6/19/07}

\title{Minimizing the Age of Information in Wireless Networks with Stochastic Arrivals}

\author{Igor Kadota}
\affiliation{%
  \institution{Massachusetts Institute of Technology}
}
\email{kadota@mit.edu}

\author{Eytan Modiano}
\affiliation{%
  \institution{Massachusetts Institute of Technology}
}
\email{modiano@mit.edu}

\renewcommand{\shortauthors}{Igor Kadota and Eytan Modiano}
\renewcommand{\shorttitle}{Minimizing the Age of Information in Wireless Networks with Stochastic Arrivals}

\begin{document}

\begin{abstract}
We consider a wireless network with a base station serving multiple traffic streams to different destinations. Packets from each stream arrive to the base station according to a stochastic process and are enqueued in a separate (per stream) queue. The queueing discipline controls which packet within each queue is available for transmission. The base station decides, at every time t, which stream to serve to the corresponding destination. The goal of scheduling decisions is to keep the information at the destinations fresh. Information freshness is captured by the Age of Information (AoI) metric.

In this paper, we derive a lower bound on the AoI performance achievable by any given network operating under any queueing discipline. Then, we consider three common queueing disciplines and develop both an Optimal Stationary Randomized policy and a Max-Weight policy under each discipline. Our approach allows us to evaluate the combined impact of the stochastic arrivals, queueing discipline and scheduling policy on AoI. We evaluate the AoI performance both analytically and using simulations. Numerical results show that the performance of the Max-Weight policy is close to the analytical lower bound.
\end{abstract}


\maketitle


\section{INTRODUCTION}\label{sec.Intro}
Traditionally, networks have been designed to maximize throughput and minimize packet latency. With the emergence of new types of networks such as vehicular networks, UAV networks and sensor networks, other performance requirements are increasingly relevant. In particular, the Age of Information (AoI) is a performance metric that was recently proposed in \cite{AoI_update,AoI_VANET} and has been receiving attention in the literature \cite{AoI_update,AoI_path,AoI_management,AoI_MG1,AoI_multiple,AoI_errors,AoI_gamma,AoI_nonlinear,Inoue17,AoI_energy15,AoI_energy17,UpdateOrWait17,AoI_lazy,igor16,AoI_scheduling,PAoI_scheduling,igorINFOCOM,Rajat_allerton,AoI_sync,Rajat17,YuPin18,YuPin17,AoI_cache,igorTON18,Vishrant17,AoI_multiaccess,AoI_discipline,AoI_LGFS17,RajatISIT18,YinSun18,BinLi18} for its application in communication systems that carry time-sensitive data. 
\emph{The AoI captures how fresh the information is from the perspective of the destination.}

Consider a system in which packets are time-stamped upon arrival. Naturally, the higher the time-stamp of a packet, the fresher its information. Let $\tau^D(t)$ be the time-stamp of the \emph{freshest packet received by the destination} by time $t$. Then, the AoI is defined as $h(t):=t-\tau^D(t)$. The AoI measures the time that elapsed since the generation of the freshest packet received by the destination. The value of $h(t)$ increases linearly over time while no fresher packet is received, representing the information getting older. At the moment a fresher packet is received, the time-stamp at the destination $\tau^D(t)$ is updated and the AoI is reduced. 

In this paper, we study a wireless network with a Base Station (BS) serving multiple traffic streams to different destinations over unreliable channels, as illustrated in Fig.~\ref{fig.Network}. Packets from each stream arrive to the BS according to a stochastic process and are enqueued in a separate (per stream) queue. The queueing discipline controls which packet within each queue is available for transmission. The BS decides, at every time $t$, which stream to serve to the corresponding destination. Our goal is to develop scheduling policies that keep the information fresh at every destination, i.e. that minimize the average AoI in the network.

In \cite{igor16}, it was shown that when the BS always has fresh packets available for transmission, the optimal scheduling policy serves the stream associated with the largest AoI. This policy is optimal\footnote{This policy was shown to minimize the average AoI of symmetric networks, i.e. networks in which all destinations have identical features.} for it gives the largest reduction in AoI over all streams. However, when packet arrivals are random, the BS may not have a fresh packet available for every stream. Thus, a scheduling policy must account both for the AoI at the destinations and the time-stamps of the packets available for transmission in each queue. For example, consider a simple network with two streams and two destinations. Assume that at time $t$, each stream has a single packet in its queue. The packet from stream $1$ was generated $30$ msecs ago and the packet from stream $2$ was generated $10$ msecs ago. Assume that the current AoI at destinations $1$ and $2$ are $h_1(t)=50$ msecs and $h_2(t)=40$ msecs, respectively. A policy that serves the stream associated with the largest AoI would select stream $1$ and yield an AoI reduction of $50-30=20$ msecs. Alternatively, serving stream $2$ would result in a reduction of $40-10=30$ msecs. Hence, to minimize the average AoI, it is optimal to schedule stream $2$. 
In this simple example, the optimal scheduling decision was easily determined. In general, designing a transmission scheduling policy that keeps information fresh over time is a challenging task that needs to take into account the packet arrival process, the queueing discipline, and the conditions of the wireless channels. 

In recent years, the problem of minimizing the AoI has been addressed in a variety of contexts. Queueing Theory is used in \cite{AoI_update,AoI_multiple,AoI_MG1,AoI_management,AoI_path,AoI_errors,AoI_gamma,AoI_nonlinear} for finding the optimal server utilization with respect to AoI. The authors in \cite{AoI_energy15,AoI_energy17,AoI_lazy,UpdateOrWait17} consider the problem of optimizing the times in which packets are generated at the source in networks with energy-harvesting or maximum update frequency constraints. Applications of AoI are studied in \cite{AoI_VANET,AoI_buffer,AoI_emulation,AoI_LUPMAC,BaiochiAoI}. Link scheduling optimization with respect to AoI has been recently considered in \cite{AoI_LGFS16,igor16,AoI_scheduling,PAoI_scheduling,igorINFOCOM,Rajat_allerton,RajatSPAWC,BrownSPAWC,AoI_sync,Rajat17,YuPin18,YuPin17,AoI_cache,igorTON18,Vishrant17,AoI_multiaccess,AoI_discipline,AoI_LGFS17,RajatISIT18,YinSun18,BinLi18}. 
Next, we describe the mentioned related work on link scheduling optimization.

The authors in \cite{AoI_LGFS17,BrownSPAWC,Rajat_allerton} studied multi-hop networks, while other works addressed single-hop networks.
Deterministic packet arrivals were considered in \cite{AoI_multiaccess,AoI_cache,Vishrant17,RajatISIT18,Rajat_allerton,Rajat17,RajatSPAWC,BrownSPAWC,igorINFOCOM,igorTON18,igor16}, arbitrary arrivals in \cite{AoI_LGFS16,AoI_LGFS17,PAoI_scheduling,AoI_scheduling,YinSun18} and stochastic arrivals in \cite{YuPin18,YuPin17,AoI_sync,BinLi18,AoI_discipline,Rajat17}. Networks with no queueing, i.e. when packets are discarded if not scheduled immediately upon arrival, were considered in \cite{YuPin17,YuPin18}, First-In First-Out (FIFO) queues were considered in \cite{PAoI_scheduling,AoI_scheduling,AoI_sync,Rajat17} and other works considered Last-Generated First-Served queues, which are often equivalent to the simpler Last-In First-Out (LIFO) queues. Reliable links over which transmissions are always successful are considered in \cite{AoI_LGFS16,AoI_LGFS17,BrownSPAWC,PAoI_scheduling,AoI_scheduling,YinSun18,YuPin18,YuPin17,AoI_sync,AoI_cache,AoI_discipline,Rajat_allerton} and other works considered unreliable links.

Most relevant to this paper are \cite{igorINFOCOM,igorTON18,AoI_sync,YinSun18,YuPin18,Rajat17}. In \cite{Rajat17}, the authors consider a network with stochastic packet arrivals, FIFO queues and link scheduling following a Stationary Randomized policy. 
An expression for the AoI in a discrete time G/Ber/1 queue is derived and used to develop a method of jointly tunning arrival and service rates of all links in order to minimize AoI. \textcolor{black}{In \cite{YinSun18}, the authors develop scheduling policies for multi-server queueng systems in which streams have synchronized packet arrivals.} In \cite{YuPin18}, the authors develop scheduling policies based on the Whittle's Index for networks with stochastic arrivals, no queues and reliable broadcast channels. The authors in \cite{AoI_sync} utilize an alternative definition of AoI to develop an Age-Based Max-Weight policy for a network with stochastic arrivals, FIFO queues and unreliable links. In \cite{igorINFOCOM,igorTON18}, the authors consider a network with deterministic arrivals, LIFO queues and unreliable broadcast channels, and develop three policies: Optimal Stationary Randomized, Whittle's Index and Age-Based Max-Weight. 

\textbf{In this paper, we develop a framework for addressing link scheduling optimization in networks with stochastic packet arrivals and unreliable links operating under three common queueing disciplines.}
Our main contributions include: i) deriving a lower bound on the AoI performance achievable by any given network operating under any queueing discipline; ii) developing both an Optimal Stationary Randomized policy and an Age-Based Max-Weight policy under three common queueing disciplines; and iii) evaluating the combined impact of the stochastic arrivals, queueing discipline and scheduling policy on AoI. We show that, contrary to intuition, the Optimal Stationary Randomized policy for LIFO queues is insensitive to packet arrival rates. Simulation results show that the performance of the Age-Based Max-Weight policy for LIFO queues is close to the analytical lower bound.

\textcolor{black}{This paper generalizes our earlier results in \cite{igorINFOCOM,igorTON18}. The main difference is that in \cite{igorINFOCOM,igorTON18} we assume that when the BS selects a stream, a new packet with fresh information is generated and then transmitted to the corresponding destination in the same time-slot. It follows that in \cite{igorINFOCOM,igorTON18} the
packet delay is always $1$ slot and the AoI is reduced to $h(t)=1$ slot after every packet delivery. In contrast, in this paper, we consider a network in which packets are generated according to a stochastic process and are enqueued before being transmitted. This seemingly modest distinction affects the packet delay and the evolution of AoI over time, which in turn affects the results and proofs throughout the paper significantly. For example, consider the analysis of Stationary Randomized policies. Under the assumptions in \cite{igorINFOCOM,igorTON18}, the AoI evolution is stochastically renewed after every packet delivery, since $h(t)=1$, and thus the AoI can be analyzed by directly applying the elementary renewal theorem for renewal-reward processes. In contrast, in this paper, the evolution of AoI may be dependent across consecutive inter-delivery intervals and, thus, the same approach is not applicable. To analyze the AoI, we obtain the stationary distribution of a two-dimensional Markov Chain in Proposition~\ref{prop.EWSAoI_randomized_OptimalBuffer}.}

The remainder of this paper is organized as follows. In Sec.~\ref{sec.Model}, we describe the network model. In Sec.~\ref{sec.Lower_Bound} we derive an analytical lower bound on the AoI minimization problem. In Sec.~\ref{sec.Randomized}, we develop the Optimal Stationary Randomized policy for each queueing discipline and characterize their AoI performance. In Sec.~\ref{sec.Max_Weight}, we develop the Max-Weight policy and obtain performance guarantees in terms of AoI. In Sec.~\ref{sec.Simulation}, we provide numerical results. 
The paper is concluded in Sec.~\ref{sec.Conclusion}. Due to the space constraint, some of the technical proofs are provided in the report in \cite{TechRep}.
\section{SYSTEM MODEL}\label{sec.Model}
Consider a wireless network with a BS serving packets from $N$ streams to $N$ destinations, as illustrated in Fig.~\ref{fig.Network}. Time is slotted with 
slot index $t\in\{1,2,\cdots,T\}$, where $T$ is the time-horizon of this discrete-time system. 
At the beginning of every slot $t$, a new packet from stream $i\in\{1,2,\cdots,N\}$ arrives to the system with probability $\lambda_i\in(0,1],\forall i$. Let $a_i(t)\in\{0,1\}$ be the indicator function that is equal to $1$ when a packet from stream $i$ arrives in slot $t$, and $a_i(t)=0$ otherwise. This Bernoulli arrival process is i.i.d. over time and independent across different streams, with $\mathbb{P}(a_i(t)=1)=\lambda_i,\forall i,t$. 

\begin{figure}[ht!]
\begin{center}
\includegraphics[width=0.9\columnwidth]{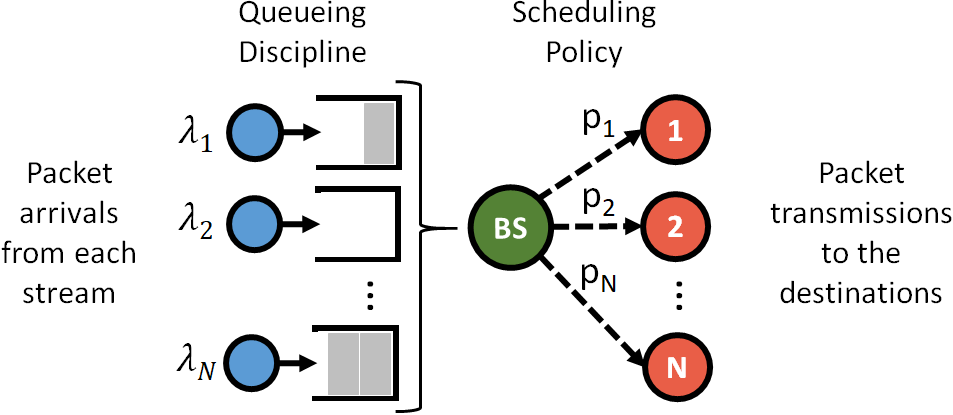}
\end{center}
\vspace{-0.2cm}
\caption{Illustration of the wireless network. 
} \label{fig.Network}
\vspace{-0.2cm}
\end{figure}

Packets from stream $i$ are enqueued in queue $i$. Denote by Head-of-Line (HoL) packets the set of packets \emph{from all queues} that are available to the BS for transmission in a given slot $t$. Depending on the queueing discipline employed by the network, queues can be of three types:
\begin{enumerate}[(i)]
\item \emph{FIFO queues}: packets are served in order of arrival. The HoL packets in slot $t$ are the oldest packets in each queue. This is a standard queueing discipline, widely deployed in communication systems. However, only a few works on link scheduling optimization \cite{PAoI_scheduling,AoI_scheduling,AoI_sync,Rajat17} consider this queueing discipline;
\item \emph{Single packet queues}: when a new packet arrives, older packets from the same stream are dropped from the queue. The HoL packets in slot $t$ are the freshest (i.e. most recently generated) packets in each queue. This queueing discipline is known to minimize the AoI in a variety of contexts. \textcolor{black}{\emph{From the perspective of the AoI, Single packet queues are equivalent to LIFO queues}};
\item \emph{No queues}: packets can be transmitted only duing the slot in which they arrive. The HoL packets in slot $t$ are given by the set $\{i|a_i(t)=1\}$. This queueing discipline is considered in \cite{YuPin17,YuPin18} for its ease of analysis.
\end{enumerate}


Let $z_i(t)$ represent the system time of the HoL packet in queue $i$ at the beginning of slot $t$. By definition, we have $z_i(t):=t-\tau_i^A(t)$, where $\tau_i^A(t)$ is the arrival time of the HoL packet in queue $i$. Naturally, the value of $\tau_i^A(t)$ changes only when the HoL packet changes, namely when the current HoL packet is served or dropped and there is another packet in the same queue; or when the queue is empty and a new packet arrives. Notice that $z_i(t)$ is undefined when queue $i$ is empty. 

We denote by $z_i^F(t)$, $z_i^S(t)$ and $z_i^N(t)$, the system times associated with \emph{FIFO queues}, \emph{Single packet queues} and \emph{No queues}, respectively. For all three cases, whenever the system time is defined, it evolves according to the definition $z_i(t):=t-\tau_i^A(t)$. Moreover, it follows from the description of the queueing disciplines that the evolution of $z_i^S(t)$ can be written as
\begin{equation}
z_i^S(t)=\left\{\begin{array}{cc} 0 & \mbox{if} \; a_i(t)=1 ; \\ z_i^S(t-1)+1 & \mbox{otherwise,}\end{array}\right. \label{eq.system_time_Optimal}
\end{equation}
and the evolution of $z_i^N(t)$ is such that $z_i^N(t)=0$ whenever an arrival occurs, i.e. $a_i(t)=1$, and is undefined otherwise. In contrast, the evolution of $z_i^F(t)$ cannot be simplified for it depends on both the arrival times and service times of packets in the queue.

In each slot $t$, the BS either idles or selects a stream and transmits its HoL packet to the corresponding destination over the wireless channel. Let $u_i(t)\in\{0,1\}$ be the indicator function that is equal to $1$ when the BS transmits the HoL packet from stream $i$ during slot $t$, and $u_i(t)=0$ otherwise. The BS can transmit at most one packet at any given time-slot $t$. Hence, we have 
\begin{equation}\label{eq.one_packet_per_slot}
\textstyle\sum_{i=1}^Nu_i(t)\leq 1, \forall t \; .
\end{equation}
The transmission scheduling policy governs the sequence of decisions $\{u_i(t)\}_{i=1}^N$ of the BS.



Let $c_i(t)\in\{0,1\}$ represent the channel state associated with destination $i$ during slot $t$. When the channel is \emph{ON}, we have $c_i(t)=1$, and when the channel is \emph{OFF}, we have $c_i(t)=0$. The channel state process is i.i.d. over time and independent across different destinations, with $\mathbb{P}(c_i(t)=1)=p_i,\forall i,t$. 

Let $d_i(t)\in\{0,1\}$ be the indicator function that is equal to $1$ when destination $i$ successfully receives a packet during slot $t$, and $d_i(t)=0$ otherwise. A successful reception occurs when the HoL packet is transmitted and the associated channel is ON, implying that $d_i(t)=c_i(t)u_i(t),\forall i,t$. 
Moreover, since the BS does not know the channel states prior to making scheduling decisions, $u_i(t)$ and $c_i(t)$ are independent, and
$\mathbb{E}[d_i(t)]=p_i\mathbb{E}[u_i(t)], \forall i,t$.
 
The transmission scheduling policies considered in this paper are non-anticipative, i.e. policies that do not use future information in making scheduling decisions. 
Let $\Pi$ be the class of non-anticipative policies and let $\pi\in\Pi$ be an arbitrary admissible policy. Our goal is to develop scheduling policies $\pi$ that minimize the average AoI in the network. Next, we formulate the AoI minimization problem.

\subsection{Age of Information}
The AoI depicts how old the information is from the perspective of the destination. Let $h_i(t)$ be the AoI associated with destination $i$ at the beginning of slot $t$. By definition, we have $h_i(t):=t-\tau_i^D(t)$, where $\tau^D_i(t)$ is the arrival time of the freshest packet delivered to destination $i$ before slot $t$. 
If during slot $t$ destination $i$ receives a packet with system time $z_i(t)=t-\tau_i^A(t)$ such that $\tau_i^A(t)>\tau_i^D(t)$, then in the next slot we have $h_i(t+1)=z_i(t)+1$. Alternatively, if during slot $t$ destination $i$ does not receive a \emph{fresher packet}, then the information gets one slot older, which is represented by $h_i(t+1)=h_i(t)+1$. 
Notice that the three queueing disciplines considered in this paper select HoL packets with increasing freshness, implying that $\tau_i^A(t)>\tau_i^D(t)$ holds\footnote{One example of a queueing discipline that can violate $\tau_i^A(t)>\tau_i^D(t)$ is the Last-In First-Out (LIFO) queue. When an older packet with $\tau_i^A(t)\leq\tau_i^D(t)$ is delivered, the associated AoI does not decrease and the network runs as if no packet was delivered. It follows that, from the perspective of the AoI, \emph{LIFO queues} are equivalent to \emph{Single packet queues}.} for every received packet.
Hence, the AoI evolves as follows:
\begin{equation}\label{eq.destination_AoI_1}
h_i(t+1)=\left\{\begin{array}{cc}z_i(t)+1 & \mbox{if} \; d_i(t)=1; \\ h_i(t)+1 & \mbox{otherwise,}\end{array}\right.
\end{equation}
for simplicity, and without loss of generality, we assume that $h_i(1)=1$ and $z_i(0)=0, \forall i$. 
Substituting $z_i^F(t)$, $z_i^S(t)$ and $z_i^N(t)$ into \eqref{eq.destination_AoI_1} we obtain the AoI associated with \emph{FIFO queues}, \emph{Single packet queues} and \emph{No queues}, respectively. In Fig.~\ref{fig.Evolution} we illustrate the evolution of $h_i(t)$ and $z_i(t)$ in a network employing \emph{Single packet queues}. 

\begin{figure}[ht!]
\begin{center}
\includegraphics[height=3.5cm]{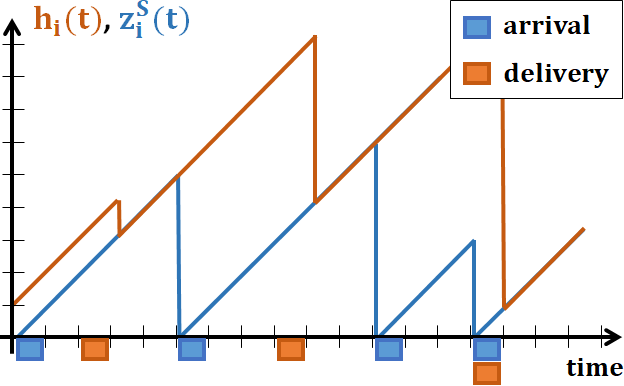}
\end{center}
\vspace{-0.2cm}
\caption{The blue and orange rectangles represent a packet arrival to queue $i$ and a successful packet delivery to destination $i$, respectively. The blue curve shows the evolution of $z_i(t)$ for the \emph{Single packet queue} and the orange curve shows the AoI associated with destination $i$.} \label{fig.Evolution}
\vspace{-0.2cm}
\end{figure}

The time-average AoI associated with destination $i$ is given by $\mathbb{E}\left[\sum_{t=1}^Th_i(t)\right]/T$. For capturing the freshness of the information of a network employing scheduling policy $\pi\in\Pi$, we define the Expected Weighted Sum AoI (EWSAoI) in the limit as the time-horizon grows to infinity as
\begin{equation}\label{eq.Objective_1}
\mathbb{E}\left[J^\pi\right]=\lim_{T\rightarrow\infty}\frac{1}{TN}\sum_{t=1}^T \sum_{i=1}^N w_i \mathbb{E}\left[h_{i}^\pi(t)\right] \; , 
\end{equation}
where $w_i$ is a positive real number that represents the priority of stream $i$. We denote by AoI-optimal, the scheduling policy $\pi^*\in\Pi$ that achieves minimum EWSAoI, namely
\begin{equation}\label{eq.AoI_optimal}
\mathbb{E}[J^*]=\min_{\pi\in\Pi}\mathbb{E}\left[J^\pi\right] \; ,
\end{equation}
where the expectation is with respect to the randomness in the channel state $c_i(t)$, scheduling decisions $u_i(t)$ and arrival process $a_i(t)$. Next, we introduce the long-term throughput and discuss the stability of FIFO queues.

\subsection{Long-term Throughput}
Let $D_i^\pi(T)=\sum_{t=1}^Td_i^\pi(t)$ be the total number of packets delivered to destination $i$ by the end of the time-horizon $T$ when the admissible policy $\pi\in\Pi$ is employed. Then, the long-term throughput associated with destination $i$ is defined as 
\begin{equation}\label{eq.throughput}
\hat{q}_i^\pi:=\lim_{T\rightarrow\infty}\frac{\mathbb{E}\left[D_i(T)\right]}{T}\; .
\end{equation}
Throughout this paper, we assume that $\hat{q}_i^\pi>0,\forall i$. Since packets from stream $i$ are generated at a rate $\lambda_i$, the long-term throughput provided to destination $i$ cannot be higher than $\lambda_i$. Hence, the long-term throughput satisfies
\begin{equation}\label{eq.throughput_condition1}
\hat{q}_i^\pi\leq \lambda_i, \forall i\; .
\end{equation}

The shared and unreliable wireless channel further restricts the set of achievable values of long-term throughput $\{\hat{q}_i^\pi\}_{i=1}^N$. By employing $\mathbb{E}[d_i(t)]=p_i\mathbb{E}[u_i(t)]$ and \eqref{eq.one_packet_per_slot} into the definition of long-term throughput in \eqref{eq.throughput}, we obtain
\begin{equation}\label{eq.throughput_condition2}
\frac{\mathbb{E}\left[D_i^\pi(T)\right]}{T}=\frac{p_i\sum_{t=1}^T\mathbb{E}[u_i^\pi(t)]}{T}\Rightarrow 
\sum_{i=1}^N\frac{\hat{q}_i^\pi}{p_i}\leq 1 \; .
\end{equation}

Inequalities \eqref{eq.throughput_condition1} and \eqref{eq.throughput_condition2} are necessary conditions\footnote{In \cite{igorINFOCOM,BinLi18}, the authors consider destinations with minimum timely-throughput requirements. Notice that conditions \eqref{eq.throughput_condition1} and \eqref{eq.throughput_condition2} are not throughput requirements enforced by the destinations. They are necessary conditions that follow naturally from the stochastic arrivals and interference constraints of the network.} for the long-term throughput $\{\hat{q}_i^\pi\}_{i=1}^N$ of any admissible scheduling policy $\pi\in\Pi$, regardless of the queueing discipline. Both inequalities are used for deriving the lower bound in Sec.~\ref{sec.Lower_Bound}. Next, we discuss the stability of FIFO queues and its impact on the AoI minimization problem. 

\subsection{Queue Stability}\label{sec.QueueStability}
Let $Q_i^\pi(t)$ be the number of packets in queue $i$ at the beginning of slot $t$ when policy $\pi$ is employed. Then, we say that queue $i$ is \emph{stable} if 
\begin{equation}\label{eq.stability}
\lim_{T\rightarrow\infty}\mathbb{E}\left[Q_i^\pi(T)\right]<\infty \; .
\end{equation}
A network is stable under policy $\pi$ when all of its queues are stable. For networks with \emph{Single packet queues} and \emph{No queues}, stability is trivial since the backlogs are such that $Q_i^\pi(t)\in\{0,1\},\forall t$, regardless of the scheduling policy. The discussion about queue stability that follows is meaningful only for the case of \emph{FIFO queues}.

\begin{defn}[Stability Region]\label{def.stability}
A set of arrival rates $\{\lambda_i\}_{i=1}^N$ is within the stability region of a given wireless network if there exists an admissible scheduling policy $\pi\in\Pi$ that stabilizes all queues.
\end{defn}

When the network is unstable under a policy $\eta\in\Pi$, then the expected backlog of at least one of its queues grows indefinitely over time. An infinitely large backlog leads to packets with infinitely large system times, i.e. $z_i(t)\rightarrow\infty$. It follows from the evolution of $h_i(t)$ in \eqref{eq.destination_AoI_1} that the AoI also increases indefinitely and, as a result, the Expected Weighted Sum AoI diverges, namely $\mathbb{E}[J^\eta]\rightarrow\infty$. Clearly, instability is a critical disadvantage for \emph{FIFO queues}. Hence, we are interested in scheduling policies that can stabilize the network whenever the arrival rates $\{\lambda_i\}_{i=1}^N$ are within the stability region. Prior to introducing the policies, we derive a lower bound to the AoI minimization problem. 

\section{LOWER BOUND}\label{sec.Lower_Bound}
In this section, we derive an alternative (and more insightful) expression for the AoI objective function $J^\pi$ in \eqref{eq.Objective_1} in terms of packet delay and inter-delivery times. Then, we use this expression to obtain a lower bound to the AoI minimization problem, namely $L_B\leq\mathbb{E}[J^*]$, for any given network operating under an arbitrary queueing discipline. Surprisingly, the lower bound $L_B$ depends only on the network's long-term throughput.

\subsection{AoI in terms of packet delay and inter-delivery times}
Consider a network employing policy $\pi$ during the time-horizon $T$. Let $\Omega$ be the sample space associated with this network and let $\omega\in\Omega$ be a sample path. For a given sample path $\omega$, let $t_i[m]$ be the index of the time-slot in which the $m$th (fresher\footnote{Recall that the delivery of an older packet with $\tau_i^A(t)\leq\tau_i^D(t)$ does not change the associated AoI and, thus, should not be counted.}) packet was delivered to destination $i$, $\forall m\in\{1,\cdots,D_i(T)\}$, where $D_i(T)$ is the total number of packets delivered. Then, we define $I_i[m]:=t_i[m]-t_i[m-1]$ as the \emph{inter-delivery time}, with $I_i[1]=t_i[1]$ and $t_i[0]=0$.

The \emph{packet delay} associated with the $m$th packet delivery to destination $i$ is given by $z_i(t_i[m])$. Notice that $z_i(t_i[m])$ is the system time of the HoL packet at the time it is delivered to the destination, which is the definition of packet delay. 
To simplify notation, we use $z_i[m]$ instead of $z_i(t_i[m])$.




Define the operator $\mathbb{\bar{M}}[\mathbf{x}]$ that calculates the sample mean of a set of values $\mathbf{x}$. Using this operator, the sample mean of $I_i[m]$ for a fixed destination $i$ is given by
\begin{equation}
\mathbb{\bar{M}}[I_i]=\frac{1}{D_i(T)}\sum_{m=1}^{D_i(T)}I_i[m] \; . \label{eq.sample_mean_I}
\end{equation}
For simplicity of notation, the time-horizon $T$ is omitted in the sample mean operator $\mathbb{\bar{M}}$.

\begin{prop}\label{prop.alternative_EWSAoI}
The infinite-horizon AoI objective function $J^\pi$ can be expressed as follows
\begin{equation}\label{eq.alternative_EWSAoI}
J^\pi=\lim_{T\rightarrow\infty}\sum_{i=1}^N\frac{w_i}{2N}\left[\frac{\mathbb{\bar{M}}[I_i^2]}{\mathbb{\bar{M}}[I_i]}+\frac{2\mathbb{\bar{M}}[z_iI_i]}{\mathbb{\bar{M}}[I_i]}+1\right] \mbox{ w.p.1} \; ,
\end{equation}
where $I_i[m]$ is the inter-delivery time, $z_i[m]$ is the packet delay and 
\begin{equation}\label{eq.sample_mean_zI}
\mathbb{\bar{M}}[z_iI_i]=\frac{1}{D_i(T)}\sum_{m=1}^{D_i(T)}z_i[m-1]I_i[m] \; . 
\end{equation}
\end{prop}

\begin{proof}
Provided in the technical report \cite[Appendix A]{TechRep}. 
\end{proof}

Equation \eqref{eq.alternative_EWSAoI} is valid for networks operating under an \emph{arbitrary queueing discipline} and employing \emph{any scheduling policy} $\pi\in\Pi$. \textcolor{black}{A similar result for the case of a single stream, $N=1$, was derived in \cite{Inoue17}}. This equation provides useful insights into the AoI minimization. The first term on the RHS of \eqref{eq.alternative_EWSAoI}, namely $\mathbb{\bar{M}}[I_i^2]/2\mathbb{\bar{M}}[I_i]$, depends only on the service regularity provided by the scheduling policy. The second term on the RHS of \eqref{eq.alternative_EWSAoI} depends on both the packet delay $z_i[m-1]$ and the inter-delivery time $I_i[m]$, as follows
\begin{equation}\label{eq.alternative_single_RHS_2}
\frac{\mathbb{\bar{M}}[z_iI_i]}{\mathbb{\bar{M}}[I_i]}=\sum_{m=1}^{D_i(T)}\frac{I_i[m]}{\sum_{j=1}^{D_i(T)}I_i[j]}z_i[m-1] \; .
\end{equation}
Notice that \eqref{eq.alternative_single_RHS_2} is a weighted sample mean of the packet delays. Intuitively, for minimizing this term, both the queueing discipline and the scheduling policy should attempt to deliver packets with low delay $z_i[m-1]$ and, when the delay is high, they should deliver the next packet as soon as possible in order to reduce the weight $I_i[m]$ on the weighted mean  \eqref{eq.alternative_single_RHS_2}. 

The expression in \eqref{eq.alternative_EWSAoI} provides intuition on how the scheduling policy should manage the packet delays $z_i[m]$ and the inter-delivery times $I_i[m]$ in order to minimize AoI. Moreover, it shows that by utilizing the simplifying assumption of queues always having fresh packets available for transmission, the scheduling policy disregards $z_i[m]$ and fails to address the term in \eqref{eq.alternative_single_RHS_2}. Next, we use \eqref{eq.alternative_EWSAoI} to obtain a lower bound to the AoI minimization problem and, in upcoming sections, we consider scheduling policies that take into account both $I_i[m]$ and $z_i[m]$.

\subsection{Lower Bound}
A lower bound on AoI is obtained from the expression in Proposition \ref{prop.alternative_EWSAoI}. By applying Jensen's inequality $\mathbb{\bar{M}}[I_i^2]\geq(\mathbb{\bar{M}}[I_i])^2$ to \eqref{eq.alternative_EWSAoI}, manipulating the resulting expression and then employing a minimization over policies in $\Pi$, we obtain
\begin{tcolorbox}[title=Lower Bound,left=1mm,right=1mm,top=-2mm,bottom=0mm]
\begin{subequations}
\begin{align}
L_B=&\min_{\pi\in\Pi}\left\{\frac{1}{2N}\sum_{i=1}^Nw_i\left(\frac{1}{\hat{q}_i^{\pi}}+1\right)\right\} \label{eq.LowerBound_1}\\
\mbox{s.t. } &\textstyle\sum_{i=1}^N \hat{q}_i^{\pi}/p_i \leq 1 \; ;\label{eq.LowerBound_2}\\
&\hat{q}_i^{\pi} \leq \lambda_i , \forall i \; ,\label{eq.LowerBound_3}
\end{align}
\end{subequations}
\end{tcolorbox}
\noindent where \eqref{eq.LowerBound_2} and \eqref{eq.LowerBound_3} are the necessary conditions for the long-term throughput in \eqref{eq.throughput_condition2} and \eqref{eq.throughput_condition1}, respectively. Notice that the optimization problem in \eqref{eq.LowerBound_1}-\eqref{eq.LowerBound_3} depends only on the network's long-term throughput $\{\hat{q}_i^{\pi}\}_{i=1}^N$ and that the condition $\hat{q}_i^{\pi} \leq \lambda_i$ limits the throughput to the packet arrival rate of the respective stream. To find the unique solution to \eqref{eq.LowerBound_1}-\eqref{eq.LowerBound_3}, we analyze the associated KKT Conditions. 

\begin{thr}[Lower bound]\label{theo.LowerBound}
For any given network with parameters $(N,p_i,\lambda_i,w_i)$ and an arbitrary queueing discipline, the optimization problem in \eqref{eq.LowerBound_1}-\eqref{eq.LowerBound_3} provides a lower bound on the AoI minimization problem, namely $L_B\leq\mathbb{E}[J^*]$. The unique solution to \eqref{eq.LowerBound_1}-\eqref{eq.LowerBound_3} is given by
\begin{equation}\label{eq.LowerBound_q}
\hat{q}_i^{L_B}=\min{\left\{\lambda_i,\sqrt{\frac{w_ip_i}{2N\gamma^*}}\right\}}, \forall i \; ,
\end{equation}
where $\gamma^*$ yields from Algorithm~\ref{alg.LowerBound}. The lower bound is given by 
\begin{equation}\label{eq.LowerBound}
L_B=\frac{1}{2N}\sum_{i=1}^Nw_i\left(\frac{1}{\hat{q}_i^{L_B}}+1\right) \; .
\end{equation}
\end{thr}

\begin{algorithm}
\caption{Solution to the Lower Bound}\label{alg.LowerBound}
\begin{algorithmic}[1]
\State $\tilde{\gamma} \gets (\sum_{i=1}^N\sqrt{w_i/p_i})^2/(2N)$ and $\gamma_i \gets w_ip_i/2N\lambda_i^2 , \forall i$
\State $\gamma \gets \max\{\tilde{\gamma};\gamma_i\}$
\State $q_i \gets \lambda_i\min\{1 ; \sqrt{\gamma_i/\gamma}\} , \forall i$
\State $S \gets \sum_{i=1}^Nq_i/p_i$
\While{$S<1$ \textbf{and} $\gamma>0$}
\State decrease $\gamma$ slightly
\State repeat steps 4 and 5 to update $q_i$ and $S$
\EndWhile
\State \textbf{return} $\gamma^*=\gamma$ and $\hat{q}_i^{L_B}=q_i,\forall i$
\end{algorithmic}
\end{algorithm}

\begin{proof}

Provided in the technical report \cite[Appendix B]{TechRep}. 
\end{proof}



Next, we develop the Optimal Stationary Randomized policy for different queueing disciplines and derive the closed-form expression for their AoI performance.
\section{STATIONARY RANDOMIZED POLICIES}\label{sec.Randomized}
Denote by $\Pi_R$ the class of Stationary Randomized policies. Let $R\in\Pi_R$ be a scheduling policy that, in each slot $t$, selects stream $i$ with probability $\mu_i\in(0,1]$ or selects no stream with probability $\mu_0$. If the selected stream $i$ has a non-empty queue, then $u_i(t)=1$ and the HoL packet is transmitted by the BS to destination $i$. Alternatively, if the selected stream $i$ has an empty queue or policy $R$ selected no stream, then $u_i(t)=0,\forall i$ and the BS idles. The scheduling probabilities $\mu_i$ are fixed over time and satisfy $\sum_{i=1}^N\mu_i=1-\mu_0$.

Randomized policies $R\in\Pi_R$ are as simple as possible. Each policy in $\Pi_R$ is fully characterized by the set $\{\mu_i\}_{i=1}^N$. They select streams at random, without taking into account $h_i(t)$, $z_i(t)$ or queue backlogs $Q_i(t)$. Notice that policies in $\Pi_R$ are not work-conserving, since they allow the BS to idle during slots in which HoL packets are available for transmission. 

Despite their simplicity, we show that by \emph{properly tuning the scheduling probabilities} $\mu_i$ according to the network parameters $(N,p_i,\lambda_i,w_i)$, policies in $\Pi_R$ can achieve performances within a factor of $4$ from the AoI-optimal. 
On the other hand, we also show that naive choices of $\mu_i$ can lead to poor AoI performances. Next, we develop and analyze scheduling policies for different queueing disciplines which are optimal over the class $\Pi_R$. In Secs.~\ref{sec.Random_OptimalBuffer}, \ref{sec.Random_NoBuffer} and \ref{sec.Random_FIFO} we consider networks employing \emph{Single packet queues}, \emph{No queues} and \emph{FIFO queues}, respectively. Then, in Sec.~\ref{sec.Computation} we compare their AoI performances.

\subsection{Randomized Policy for Single packet queue}\label{sec.Random_OptimalBuffer}
Consider a network employing the \emph{Single packet queue} discipline on $N$ streams with packet arrival rates $\lambda_i$, priorities $w_i$ and channel reliabilities $p_i$. Recall that for the \emph{Single packet queue}, when a new packet arrives, older packets from the same stream are dropped. The BS selects streams according to $R\in\Pi_R$ with scheduling probabilities $\mu_i$. Following a successful packet transmission from stream $i$, its queue remains empty or a new packet arrives. The expected number of (consecutive) slots that queue $i$ remains empty is $1/\lambda_i-1$. When a new packet arrives, the BS transmits this packet with probability $\mu_i$. The expected number of slots necessary to successfully deliver this packet is $1/p_i\mu_i$. Under policy $R\in\Pi_R$ and for the case of \emph{Single packet queues}, the sequence of packet deliveries is a renewal process. It follows from the elementary renewal theorem \cite{DSP} that
\begin{equation}
\lim_{T\rightarrow\infty}\frac{1}{T}\sum_{t=1}^T\mathbb{E}[d_i(t)]=\frac{1}{1/p_i\mu_i+1/\lambda_i-1},\forall i,t \; .
\end{equation}

For the particular case of $\lambda_i=1$, the AoI process $h_i(t)$ is also stochastically renewed after every packet delivery and the long-term time-average $\mathbb{E}[h_i(t)]$ can be easily obtained using the elementary renewal theorem for renewal-reward processes. In contrast, for the general case of $\lambda_i\in(0,1]$, the evolution of $h_i(t)$ may be dependent across consecutive inter-delivery intervals due to its relationship with the system time $z_i^S(t)$ given in \eqref{eq.destination_AoI_1}. 
To find an expression for the long-term time-average $\mathbb{E}[h_i(t)]$ we formulate the problem as a two-dimensional Markov Chain with countably-infinite state space represented by $(h_i(t),z_i(t))$ and obtain its stationary distribution. Proposition~\ref{prop.EWSAoI_randomized_OptimalBuffer} follows from substituting the expression for $\mathbb{E}[h_i(t)]$ into the objective function in \eqref{eq.AoI_optimal}.

\begin{prop}\label{prop.EWSAoI_randomized_OptimalBuffer}
The optimal EWSAoI achieved by a network with Single packet queues over the class $\Pi_R$ is given by
\begin{tcolorbox}[title=Optimal Randomized policy for Single packet queues,left=1mm,right=1mm,top=-2mm,bottom=0mm]
\begin{subequations}
\begin{align}
\mathbb{E}\left[J^{R^S}\right]=&\min_{R\in\Pi_R}\left\{\frac{1}{N} \sum_{i=1}^N w_i \left(\frac{1}{\lambda_i}-1+\frac{1}{p_i\mu_i}\right)\right\} \label{eq.EWSAoI_randomized_OptimalBuffer}\\
\mbox{s.t. } &\textstyle\sum_{i=1}^N \mu_i \leq 1 \; ;\label{eq.EWSAoI_randomized_OptimalBuffer_2}
\end{align}
\end{subequations}
\end{tcolorbox}
\noindent where $R^S$ denotes the Optimal Stationary Randomized Policy for the Single packet queue discipline.
\end{prop}
\begin{proof}
Provided in the technical report \cite[Appendix C]{TechRep}. 
\end{proof}

Next, we solve the optimization problem in \eqref{eq.EWSAoI_randomized_OptimalBuffer}-\eqref{eq.EWSAoI_randomized_OptimalBuffer_2} and obtain the optimal scheduling probabilities $\{\mu_i^S\}_{i=1}^N$.

\begin{thr}\label{theo.scheduling_policy_OptimalBuffer}
Consider a network with parameters $(N,p_i,\lambda_i,w_i)$ operating under the Single packet queues discipline. The optimal scheduling probabilities are given by
\begin{equation}\label{eq.scheduling_policy_OptimalBuffer}
\mu_i^S=\frac{\sqrt{w_i/p_i}}{\sum_{j=1}^N\sqrt{w_j/p_j}} , \forall i \; ,
\end{equation}
and the performance of the Optimal Stationary Randomized policy $R^S$ is
\begin{equation}\label{eq.EWSAoI_randomized_OptimalBuffer_final}
\mathbb{E}\left[J^{R^S}\right]=\frac{1}{N} \sum_{i=1}^N w_i \left(\frac{1}{\lambda_i}-1\right)+\frac{1}{N}\left( \sum_{i=1}^N \sqrt{\frac{w_i}{p_i}}\right)^2 .
\end{equation}
Then, it follows that 
\begin{equation}\label{eq.Rand_OptimalBuffer}
\mathbb{E}\left[J^*\right]\leq \mathbb{E}\left[J^{R^S}\right]< 4\mathbb{E}\left[J^*\right] \; ,
\end{equation}
where $\mathbb{E}\left[J^*\right]=\min_{\pi\in\Pi}\mathbb{E}\left[J^\pi\right]$ is the minimum AoI over the class of all admissible policies $\Pi$.
\end{thr}
\begin{proof}
The scheduling probabilities $\{\mu_i^S\}_{i=1}^N$ that minimize \eqref{eq.EWSAoI_randomized_OptimalBuffer}-\eqref{eq.EWSAoI_randomized_OptimalBuffer_2} also minimize this equivalent problem
\begin{equation}\label{eq.EWSAoI_randomized_OptimalBuffer_equiv}
\min_{R\in\Pi_R}\left\{\frac{1}{N} \sum_{i=1}^N \frac{w_i}{p_i\mu_i}\right\} \; \mbox{ s.t. } \sum_{i=1}^N \mu_i \leq 1 \; .
\end{equation}
Consider the Cauchy - Schwarz inequality
\begin{equation}
\left(\sum_{i=1}^N \sqrt{\frac{w_i}{p_i}}\right)^2\leq\left(\sum_{i=1}^N \mu_i\right)\left(\sum_{i=1}^N \frac{w_i}{p_i\mu_i}\right) \; .
\end{equation}
The LHS is a lower bound on the objective function in \eqref{eq.EWSAoI_randomized_OptimalBuffer_equiv}. Notice that Cauchy - Schwarz holds with equality when $\{\mu_i^S\}_{i=1}^N$ is given by \eqref{eq.scheduling_policy_OptimalBuffer}, implying that \eqref{eq.scheduling_policy_OptimalBuffer} is a solution to both \eqref{eq.EWSAoI_randomized_OptimalBuffer_equiv} and \eqref{eq.EWSAoI_randomized_OptimalBuffer}-\eqref{eq.EWSAoI_randomized_OptimalBuffer_2}. Substituting the solution\footnote{The expression in \eqref{eq.scheduling_policy_OptimalBuffer} was obtained in previous work \cite{igorTON18} under the simplifying assumption of all streams always having fresh packets available for transmission. In Theorem~\ref{theo.scheduling_policy_OptimalBuffer} we show that \eqref{eq.scheduling_policy_OptimalBuffer} is in fact optimal for streams with stochastic packet arrivals and for any set of arrival rates $\{\lambda_i\}_{i=1}^N$.} $\{\mu_i^S\}_{i=1}^N$ into the objective function in \eqref{eq.EWSAoI_randomized_OptimalBuffer} gives \eqref{eq.EWSAoI_randomized_OptimalBuffer_final}. 

For deriving the upper bound in \eqref{eq.Rand_OptimalBuffer}, consider the Randomized policy $\tilde R$ with $\tilde{\mu}_i=\hat{q}_i^{L_B}/p_i,\forall i$. Substitute $\tilde{\mu}_i$ into the RHS of \eqref{eq.EWSAoI_randomized_OptimalBuffer} and denote the result as $\mathbb{E}[J^{\tilde R}]$. Comparing $L_B$ in \eqref{eq.LowerBound} with $\mathbb{E}[J^{\tilde R}]$ and noting from \eqref{eq.LowerBound_q} that $\hat{q}_i^{L_B}\leq \lambda_i$, gives that 
\begin{equation}\label{eq.Rand_OptimalBuffer_1}
\mathbb{E}\left[J^{\tilde R}\right]\leq \frac{1}{N} \sum_{i=1}^N w_i \left(\frac{2}{p_i\tilde{\mu}_i}-1\right)<4L_B \; .
\end{equation}
By definition, we know that
\begin{equation}\label{eq.Rand_OptimalBuffer_2}
L_B\leq\mathbb{E}[J^*]\leq\mathbb{E}[J^{R^S}]\leq\mathbb{E}[J^{\tilde R}] \; .
\end{equation}
Inequality \eqref{eq.Rand_OptimalBuffer} follows directly from \eqref{eq.Rand_OptimalBuffer_1} and \eqref{eq.Rand_OptimalBuffer_2}.
\end{proof}



\textbf{Intuitively, the optimal probabilities $\mathbf{\{\mu_i\}_{i=1}^N}$ should vary with the packet arrival rates $\mathbf{\{\lambda_i\}_{i=1}^N}$}. For example, consider a \emph{Single packet queue} with low arrival rate and high scheduling probability. This queue is often offered service while empty, thus wasting resources. Hence, it seems natural that the optimal $\mu_i$ should  vary with $\lambda_i$. In Secs.~\ref{sec.Random_NoBuffer} and \ref{sec.Random_FIFO}, we show that this is the case for \emph{No queues} and \emph{FIFO queues}. However, Theorem~\ref{theo.scheduling_policy_OptimalBuffer} shows that \textbf{for \emph{Single packet queues} the optimal $\mathbf{\mu_i^S}$ depends only on $\mathbf{w_i}$ and $\mathbf{p_i}$. This result is important for it simplifies the design of networked systems that attempt to minimize AoI}, as discussed in Sec.~\ref{sec.Computation}.



\subsection{Randomized Policy for No queue}\label{sec.Random_NoBuffer}
Consider a network with parameters $(N,p_i,\lambda_i,w_i)$ employing the \emph{No queue} discipline and a Stationary Randomized policy $R\in\Pi_R$ with scheduling probabilities $\mu_i$. Recall that $R$ is oblivious to packet arrivals and that, under the \emph{No queue} discipline, packets are available for transmission only during the slot in which they arrive to the system. Hence, if $R$ selects stream $i$ during slot $t$, a successful packet delivery occurs only if a packet from stream $i$ arrived at the beginning of slot $t$, i.e. $a_i(t)=1$, and the channel is ON, i.e. $c_i(t)=1$. Therefore, for the \emph{No queue} discipline, we have that  $d_i(t)=a_i(t)c_i(t)u_i(t),\forall i,t$. This is equivalent to a network with a \emph{virtual channel} that is ON with probability $p_i\lambda_i$ and OFF with probability $1-p_i\lambda_i$. We use this equivalence to derive the results that follow. 

\begin{prop}\label{prop.EWSAoI_randomized_NoBuffer}
The optimal EWSAoI achieved by a network with No queues over the class $\Pi_R$ is given by
\begin{tcolorbox}[title=Optimal Randomized policy for No queues,left=1mm,right=1mm,top=-2mm,bottom=0mm]
\begin{subequations}
\begin{align}
\mathbb{E}\left[J^{R^N}\right]=&\min_{R\in\Pi_R}\left\{\frac{1}{N} \sum_{i=1}^N \frac{w_i}{p_i\mu_i\lambda_i}\right\}\label{eq.EWSAoI_randomized_NoBuffer}\\
\mbox{s.t. } &\textstyle\sum_{i=1}^N \mu_i \leq 1 \; ;\label{eq.EWSAoI_randomized_NoBuffer_2}
\end{align}
\end{subequations}
\end{tcolorbox}
where $R^N$ denotes the Optimal Stationary Randomized policy for the No queues discipline.
\end{prop}
\begin{proof}
Under the \emph{No queues} discipline, all packets are delivered with system time $z_i^N(t)=0$ and the AoI process $h_i(t)$ is renewed after every packet delivery. Hence, it follows from the elementary renewal theorem for renewal-reward processes that
\begin{equation}\label{eq.h_Noqueue}
\lim_{T\rightarrow\infty}\frac{1}{T}\sum_{t=1}^T\mathbb{E}[h_i(t)]=\frac{1}{p_i\mu_i\lambda_i}\; .
\end{equation}
Substituting \eqref{eq.h_Noqueue} into \eqref{eq.AoI_optimal} gives \eqref{eq.EWSAoI_randomized_NoBuffer}.
\end{proof}

\begin{thr}\label{theo.scheduling_policy_NoBuffer}
Consider a network with parameters $(N,p_i,\lambda_i,w_i)$ operating under the No queues discipline. The optimal scheduling probabilities are given by
\begin{equation}\label{eq.scheduling_policy_NoBuffer}
\mu_i^N=\frac{\sqrt{w_i/p_i\lambda_i}}{\sum_{j=1}^N\sqrt{w_j/p_j\lambda_j}} , \forall i \; ,
\end{equation}
and the performance of the Optimal Stationary Randomized policy $R^N$ is
\begin{equation}\label{eq.EWSAoI_randomized_NoBuffer_final}
\mathbb{E}\left[J^{R^N}\right]=\frac{1}{N}\left( \sum_{i=1}^N \sqrt{\frac{w_i}{p_i\lambda_i}}\right)^2 .
\end{equation}
\end{thr}
\begin{proof}
The proof is similar to Theorem~\ref{theo.scheduling_policy_OptimalBuffer}.
\end{proof}

As expected, the similarities between the Optimal Stationary Randomized policies for the \emph{No queue} and \emph{Single packet queue} disciplines increase as the packet arrival rates $\{\lambda_i\}_{i=1}^N$ increase. In particular, notice from \eqref{eq.scheduling_policy_OptimalBuffer} and \eqref{eq.scheduling_policy_NoBuffer} that $\mu_i^N=\mu_i^S,\forall i$, when $\lambda_i=1,\forall i$, and, as a result, their AoI performance is also identical, namely $\mathbb{E}\left[J^{R^N}\right]=\mathbb{E}\left[J^{R^S}\right]$ when $\lambda_i=1,\forall i$. Recall that $\mu_i^S$ does not change with $\lambda_i$.

\subsection{Randomized Policy for FIFO queue}\label{sec.Random_FIFO}
Consider a network with parameters $(N,p_i,\lambda_i,w_i)$ employing \emph{FIFO queues} and a Stationary Randomized policy $R\in\Pi_R$ with scheduling probabilities $\mu_i$. In this setting, each \emph{FIFO queue} behaves as a \emph{discrete-time Ber/Ber/1 queue} with arrival rate $\lambda_i$ and service rate $p_i\mu_i$. From \cite[Sec.~8.10]{Mor_book}, we know that the \emph{FIFO queue} is \emph{stable} when $p_i\mu_i>\lambda_i$ and that its steady-state expected backlog is given by
\begin{equation}\label{eq.backlogFIFO}
\lim_{T\rightarrow\infty}\mathbb{E}\left[Q_i(T)\right]=\frac{\lambda_i(1-p_i\mu_i)}{p_i\mu_i-\lambda_i} \; .
\end{equation}
From \cite[Theorem 5]{Rajat17}\footnote{The authors in \cite{Rajat17} obtain the minimum value of \eqref{eq.Rand_FIFOBuff_1} by jointly optimizing over scheduling probabilities $\{\mu_i^F\}_{i=1}^N$ and packet arrival rates $\{\lambda_i\}_{i=1}^N$. Theorem~\ref{theo.scheduling_policy_FIFOBuffer} generalizes this result, by providing the optimal $\{\mu_i^F\}_{i=1}^N$ for any given $\{\lambda_i\}_{i=1}^N$.}, we know that the AoI associated with a \emph{stable FIFO queue} is given by
\begin{equation}\label{eq.Rajat_paper}
\lim_{T\rightarrow\infty}\frac{1}{T}\sum_{t=1}^{T}\mathbb{E}[h_i(t)]=\frac{1}{p_i\mu_i}+\frac{1}{\lambda_i}+\left[\frac{\lambda_i}{p_i\mu_i}\right]^2\frac{1-p_i\mu_i}{p_i\mu_i-\lambda_i} \; .
\end{equation}
Notice the similarities between \eqref{eq.Rajat_paper}, the expected backlog in \eqref{eq.backlogFIFO} and the AoI associated with a \emph{Single packet queue} in \eqref{eq.EWSAoI_randomized_OptimalBuffer}. Under light load, i.e. when $\lambda_i<<p_i\mu_i$, the third term on the RHS of \eqref{eq.Rajat_paper} is small when compared to the other terms. Hence, the AoI of the \emph{FIFO queue} in \eqref{eq.Rajat_paper} is similar to the AoI of the \emph{Single packet queue} in \eqref{eq.EWSAoI_randomized_OptimalBuffer}. On the other hand, under heavy load, as $\lambda_i\rightarrow p_i\mu_i$, the third term on the RHS of \eqref{eq.Rajat_paper} dominates. Both the backlog and the AoI of the \emph{FIFO queue}, in \eqref{eq.backlogFIFO} and \eqref{eq.Rajat_paper}, respectively, increase sharply. Recall that when the backlog is large, packets have to wait for a long time in the queue before being served, what makes their information stale and, as a result, the AoI large. The \emph{Single packet queue} discipline avoids this issue by keeping only the freshest packet in the queue.


Denote by $R^F$ the Optimal Stationary Randomized policy for the case of \emph{FIFO queues} and let $\{\mu_i^F\}_{i=1}^N$ be the associated scheduling probabilities. Substituting \eqref{eq.Rajat_paper} into the expression for the EWSAoI in \eqref{eq.AoI_optimal} gives
\begin{tcolorbox}[title=Optimal Randomized policy for FIFO queues,left=1mm,right=1mm,top=-2mm,bottom=0mm]
\begin{subequations}
\begin{align}
\mathbb{E}\left[J^{R^F}\right]=&\min_{R\in\Pi_R}\left\{\sum_{i=1}^N \frac{w_i}{N} \left[\frac{1}{p_i\mu_i}+\frac{1}{\lambda_i}+\right.\right.\nonumber\\
&\phantom{xxxxxxxxx}\left.\left.+\left[\frac{\lambda_i}{p_i\mu_i}\right]^2\frac{1-p_i\mu_i}{p_i\mu_i-\lambda_i}\right]\right\}\label{eq.Rand_FIFOBuff_1}\\
\mbox{s.t. } &\textstyle\sum_{i=1}^N \mu_i \leq 1 \; ; \label{eq.Rand_FIFOBuff_2}\\
&p_i\mu_i > \lambda_i , \forall i \; .\label{eq.Rand_FIFOBuff_3}
\end{align}
\end{subequations}
\end{tcolorbox}
where \eqref{eq.Rand_FIFOBuff_2} is the constraint on scheduling decisions and \eqref{eq.Rand_FIFOBuff_3} is the condition for network stability.



\begin{remk}\label{rem.stability}
A sufficient condition for $\{\lambda_i\}_{i=1}^N$ to be within the stability region of the network is given by $\sum_{i=1}^N\lambda_i/p_i < 1$.
\end{remk}

\begin{thr}\label{theo.scheduling_policy_FIFOBuffer}
The optimal scheduling probabilities for the case of FIFO queues $\mu_i^F$ are given by Algorithm~\ref{alg.FIFO} when $\delta\rightarrow 0$.
\end{thr}

\begin{proof}
The auxiliary parameter $\delta>0$ is used to enforce a closed feasible set to the optimization problem in \eqref{eq.Rand_FIFOBuff_1}-\eqref{eq.Rand_FIFOBuff_3}. We exchange \eqref{eq.Rand_FIFOBuff_3} by $p_i\mu_i \geq \lambda_i +\delta,\forall i$, to ensure that Algorithm~\ref{alg.FIFO} always finds a unique solution to the KKT Conditions associated with \eqref{eq.Rand_FIFOBuff_1}-\eqref{eq.Rand_FIFOBuff_3} for any fixed (and arbitrarily small) value of $\delta$. Recall that when $p_i\mu_i \approx \lambda_i$ the AoI performance is poor. Hence, in most cases, the optimal scheduling probabilities $\{\mu_i^F\}_{i=1}^N$ are such that $p_i\mu_i^F$ and $\lambda_i$ are not close, meaning that small changes in $\delta$ should not affect the solution. Algorithm \ref{alg.FIFO} finds the unique solution to the KKT Conditions and is developed using a similar method as in Theorem~\ref{theo.LowerBound}.
\end{proof}
As part of Algorithm~\ref{alg.FIFO}, we use the partial derivative of \eqref{eq.Rajat_paper} with respect to $\mu_i$ multiplied by $w_i/N$, which is denoted as 
\begin{equation}\label{eq.g}
g_i(x)=\frac{w_i}{N}\left\{\frac{\lambda_i}{p_i\mu_i^2}\left[\frac{2}{p_i\mu_i}-1\right]-\frac{p_i(1-\lambda_i)}{(p_i\mu_i-\lambda_i)^2}\right\}_{x=\mu_i}
\end{equation}

\begin{algorithm}
\caption{Randomized policy for FIFO queue}\label{alg.FIFO}
\begin{algorithmic}[1]
\State $\gamma_i \gets (\lambda_i+\delta)/p_i \; , \forall i \in \{1,2,\cdots,N\}$
\State $\gamma \gets \max_{i}\{-g_i(\gamma_i)\}$ \Comment{where $g_i(.)$ is given in \eqref{eq.g}}
\State $\mu_i \gets \max\{\;\gamma_i\;;\;g_i^{-1}(-\gamma)\;\}$
\State $S \gets \mu_1+\mu_2+\cdots+\mu_N$
\While{$S<1$}
\State decrease $\gamma$ slightly
\State repeat steps 3 and 4 to update $\mu_i$ and $S$
\EndWhile
\State \textbf{return} $\mu_i^F=\mu_i,\forall i$
\end{algorithmic}
\end{algorithm}



\subsection{Comparison of Queueing Disciplines}\label{sec.Computation}
Next, we compare the performance of four different Stationary Randomized Policies: 1) Optimal Policy for \emph{Single packet queues}, $R^S$; 2) Optimal Policy for \emph{No queues}, $R^N$; 3) Optimal Policy for \emph{FIFO queues}, $R^F$; and 4) Naive Policy for \emph{FIFO queues}. The EWSAoI of the first three policies is computed using \eqref{eq.EWSAoI_randomized_OptimalBuffer_final}, \eqref{eq.EWSAoI_randomized_NoBuffer_final} and the solution to \eqref{eq.Rand_FIFOBuff_1}-\eqref{eq.Rand_FIFOBuff_3}, respectively. The Naive Policy shares resources evenly between streams by assigning $\mu_i=1/N,\forall i$. The EWSAoI of the Naive Policy is computed using the expression inside the minimization in \eqref{eq.Rand_FIFOBuff_1}.

We consider a network with two streams, $w_1=w_2=1$, $p_1=1/3$, $p_2=1$, $\lambda_1=\lambda$, $\lambda_2=\lambda/3$ and varying arrival rates $\lambda\in\{0.01,0.02,\cdots,1\}$. In Fig.~\ref{fig.Rand_Computation}, we show the EWSAoI of Randomized Policies under different queueing disciplines and display the Lower Bound $L_B$ for comparison. The policy with \emph{Single packet queues} outperforms the policies with other queueing disciplines for every arrival rate $\lambda$, as expected. 





\begin{figure}[ht!]
\begin{center}
\includegraphics[height=5cm]{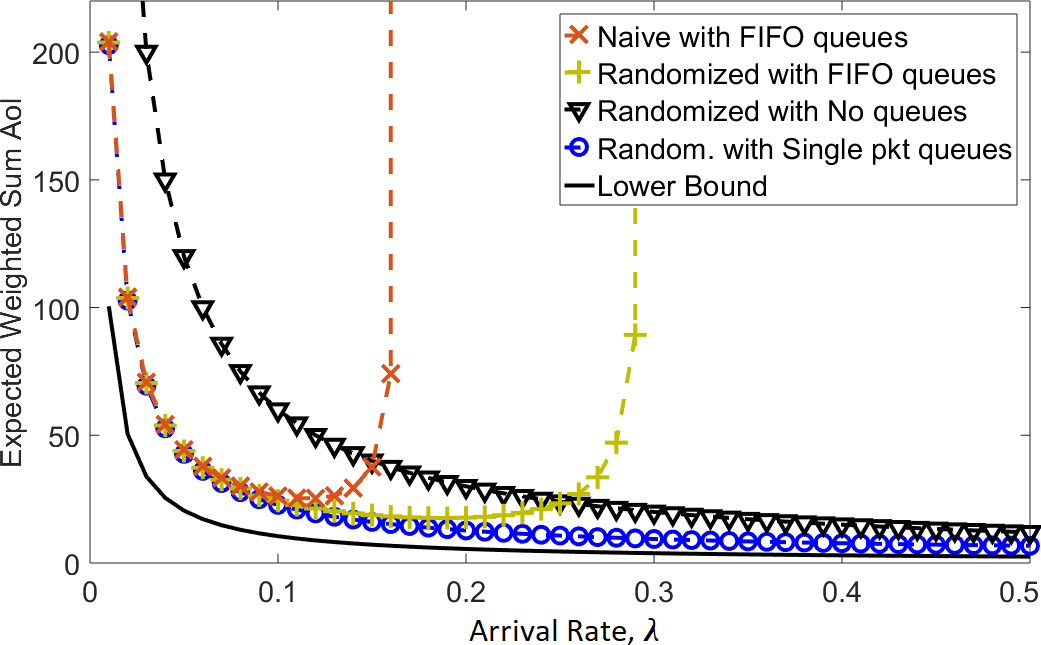}
\end{center}
\vspace{-0.2cm}
\caption{Comparison of Stationary Randomized Policies.} \label{fig.Rand_Computation}
\vspace{-0.2cm}
\end{figure}

The Optimal Policy for \emph{FIFO queues} leverages its knowledge of $p_i$ and $\lambda_i$ to stabilize the network whenever $\{\lambda_i\}_{i=1}^N$ is within the stability region. In contrast, the Naive Policy shares channel resources evenly between streams, disregarding queue stability. From Remark~\ref{rem.stability}, we know that the network can be stabilized for $\lambda<3/10$. However, in Fig.~\ref{fig.Rand_Computation}, we observe that the Naive Policy is unable to stabilize the network when $\lambda\in(1/6,3/10)$. By comparing their performances, it becomes evident that stability is critical for \emph{FIFO queues}. 

Both the \emph{Single packet queue} and the \emph{No queue} disciplines present a natural relationship between the rate at which fresh information is generated at the source $\lambda_i$ and the resulting AoI at the destination, namely a higher arrival rate (always) leads to a lower AoI. Furthermore, Theorem~\ref{theo.scheduling_policy_OptimalBuffer} shows that the optimal scheduling probabilities $\mu_i^S$ for \emph{Single packet queues} are independent of $\lambda_i$. \textbf{This result allows us to isolate the design of the arrival rate $\mathbf{\lambda_i}$ from the design of the scheduling probability $\mathbf{\mu_i}$}. In particular, to minimize the EWSAoI in the network, the arrival rates $\{\lambda_i\}_{i=1}^N$ should be set as high as possible, while the scheduling probabilities $\{\mu_i^S\}_{i=1}^N$ should be proportional to $\sqrt{w_i/p_i}$ according to \eqref{eq.scheduling_policy_OptimalBuffer}. \textbf{Since arrival rates and scheduling policies are often defined by different layers of the network stack, this isolation simplifies the design of networked systems. It is important to emphasize that this isolation only holds for networks employing \emph{Single packet queues}. For \emph{FIFO queues} and \emph{No queues} the optimal value of $\mathbf{\mu_i}$ changes for different values of $\mathbf{\lambda_i}$}. Next, we develop Age-Based Max-Weight Policies that use the knowledge of $h_i(t)$ and $z_i(t)$ for making scheduling decisions in an adaptive manner.
\section{AGE-BASED MAX-WEIGHT POLICIES}\label{sec.Max_Weight}
In this section, we use Lyapunov Optimization \cite{lyapunov} to develop Age-Based Max-Weight policies for each of the queueing disciplines. The Max-Weight policy is designed to reduce the expected drift of the Lyapunov Function at every slot $t$. In doing so, the Max-Weight policy attempts to minimize the AoI of the network.

We use the following linear Lyapunov Function
\begin{equation}\label{eq.LyapunovFunction}
L\left(\{h_i(t)\}_{i=1}^N\right)=L(t)=\frac{1}{N}\sum_{i=1}^N \beta_ih_i(t) \; ,
\end{equation}
where $\beta_i$ is a positive hyperparameter that can be used to tune the Max-Weight policy to different network configurations and queueing disciplines. The Lyapunov Drift is defined as
\begin{equation}\label{eq.LyapunovDrift}
\Delta(\mathbb{S}(t)):=\mathbb{E}\left[ \left. L(t+1) - L(t)\right| \mathbb{S}(t) \right] \; ,
\end{equation}
where $\mathbb{S}(t)=(\{h_i(t)\}_{i=1}^N,\{z_i(t)\}_{i=1}^N)$ is the network state at the beginning of time slot $t$.
The Lyapunov Function $L(t)$ increases with the AoI of the network and the Lyapunov Drift $\Delta(\mathbb{S}(t))$ represents the expected increase of $L(t)$ in one slot. Hence, by minimizing the drift in \eqref{eq.LyapunovDrift} at every slot $t$, the Max-Weight policy is attempting to keep both $L(t)$ and the network's AoI small. 

To develop the Max-Weight policy, we analyze the expression for the drift in \eqref{eq.LyapunovDrift}. Substituting the evolution of $h_i(t+1)$ from \eqref{eq.destination_AoI_1} into \eqref{eq.LyapunovDrift} and then manipulating the resulting expression, we obtain
\begin{align}\label{eq.LyapunovDrift_1}
\Delta(\mathbb{S}(t))&=\frac{1}{N}\sum_{i=1}^N\beta_i-\frac{1}{N}\sum_{i=1}^N\beta_ip_i\left(h_i(t)-z_i(t)\right)\mathbb{E}\left[ \left. u_i(t) \right| \mathbb{S}(t) \right]\; . 
\end{align}
The scheduling decision in slot $t$ affects only the second term on the RHS of \eqref{eq.LyapunovDrift_1}. For minimizing $\Delta(\mathbb{S}(t))$, the \emph{Max-Weight policy selects, in each slot $t$, the stream $i$ with a HoL packet and the highest value of} $\beta_ip_i\left(h_i(t)-z_i(t)\right)$, with ties being broken arbitrarily. The Max-Weight policy is work-conserving since it idles only when all queues are empty. 

Substituting $z_i^S(t)$, $z_i^N(t)$ and $z_i^F(t)$ into $\beta_ip_i\left(h_i(t)-z_i(t)\right)$ gives the Max-Weight policy associated with the \emph{Single packet queue}, $MW^S$, the \emph{No queue}, $MW^N$, and the \emph{FIFO queue}, $MW^F$, respectively. Notice that the difference $h_i(t)-z_i(t)$ represents the AoI reduction accrued from a successful packet delivery to destination $i$. Hence, it makes sense that the Max-Weight policy prioritizes queues with high potential reward $h_i(t)-z_i(t)$. 

\begin{thr}[Performance Bounds for $MW^S$]\label{theo.MW_bound_Optimal}
Consider a network employing Single packet queues. The performance of the Max-Weight policy with $\beta_i=w_i/p_i\mu_i^S,\forall i$, is such that
\begin{equation}\label{eq.MW_bound_Optimal}
\mathbb{E}\left[J^{MW^S}\right]\leq\mathbb{E}\left[J^{R^S}\right]\; ,
\end{equation}
where $\mu_i^S$ and $\mathbb{E}[J^{R^S}]$ are the optimal scheduling probability for the case of Single packet queues and the associated EWSAoI attained by $R^S$, respectively.
\end{thr}
%

\begin{thr}[Performance Bounds for $MW^N$]\label{theo.MW_bound_NoBuffer}
Consider a network employing the No queues discipline. The performance of the Max-Weight Policy with $\beta_i=w_i/p_i\mu_i^N,\forall i$, is such that
\begin{equation}\label{eq.MW_bound_NoBuffer}
\mathbb{E}\left[J^{MW^N}\right]\leq\mathbb{E}\left[J^{R^N}\right]  \; ,
\end{equation}
where $\mu_i^N$ and $\mathbb{E}[J^{R^N}]$ are the optimal scheduling probability for the case of No queues and the associated EWSAoI attained by $R^N$, respectively.
\end{thr}

The proofs of Theorems~\ref{theo.MW_bound_Optimal} and \ref{theo.MW_bound_NoBuffer} are provided in the technical report \cite[Appendices~D and E]{TechRep}, respectively. Both proofs rely on the construction of equivalent systems that facilitate the analysis of the expression of the drift in \eqref{eq.LyapunovDrift_1}. The performance of $MW^F$ is evaluated next using simulations. 

Stationary Randomized policies select streams randomly, according to a fixed set of scheduling probabilities $\{\mu_i\}_{i=1}^N$. In contrast, Max-Weight policies leverage the knowledge of $h_i(t)$ and $z_i(t)$ to select which stream to serve. Therefore, it is not surprising that Max-Weight policies outperform Randomized policies. However, establishing a performance guarantee as in  \eqref{eq.MW_bound_Optimal} and \eqref{eq.MW_bound_NoBuffer} is challenging for it depends on finding a tight upper bound for the performance of Max-Weight policies, which often do not have properties such as \emph{renewal intervals} that simplify the analysis.
Next, we provide numerical results that further validate the superior performance of the Max-Weight policies.



\section{NUMERICAL RESULTS}\label{sec.Simulation}

In this section, we evaluate the performance of scheduling policies in terms of the EWSAoI. We compare: i) the Optimal Stationary Randomized Policy for the case of \emph{Single packet queues} $R^S$, \emph{No queues} $R^N$ and \emph{FIFO queues} $R^F$; ii) the Max-Weight Policy\footnote{For the Max-Weight Policies $MW^S$, $MW^N$ and $MW^F$, we employ $\beta_i=w_i/p_i\mu_i^X,\forall i$, where $\mu_i^X$ is the optimal scheduling probability for the associated queueing discipline.} for the case of \emph{Single packet queues} $MW^S$, \emph{No queues} $MW^N$ and \emph{FIFO queues} $MW^F$; and iii) the Whittle's Index Policy under the \emph{No queues} discipline. The first two policies were developed in Secs.~\ref{sec.Randomized} and \ref{sec.Max_Weight}, respectively, and the last policy was proposed in \cite{YuPin18}. The Lower Bound $L_B$ derived in Sec.~\ref{sec.Lower_Bound} is displayed for comparison.

In Figs.~\ref{fig.Var_A} and \ref{fig.Var_A_1}, we simulate networks with time-horizon $T=2\times 10^6$ slots and $N=4$ traffic streams with priorities $w_1=4$, $w_2=4$, $w_3=1$, $w_4=1$, channel reliabilities $p_i=i/N, \forall i$ and arrival rates $\lambda_i=(N-i+1)/N\times \lambda$ for $\lambda\in\{0.01,0.02,\cdots,0.35\}$. The results are separated in two figures for clarity. The performance of the Randomized policies is computed using the expressions in Sec.~\ref{sec.Randomized} while the performance of the Max-Weight and Whittle's Index policies are averages over $10$ simulation runs.
    
\begin{figure}[ht!]
\begin{center}
\includegraphics[height=5.3cm]{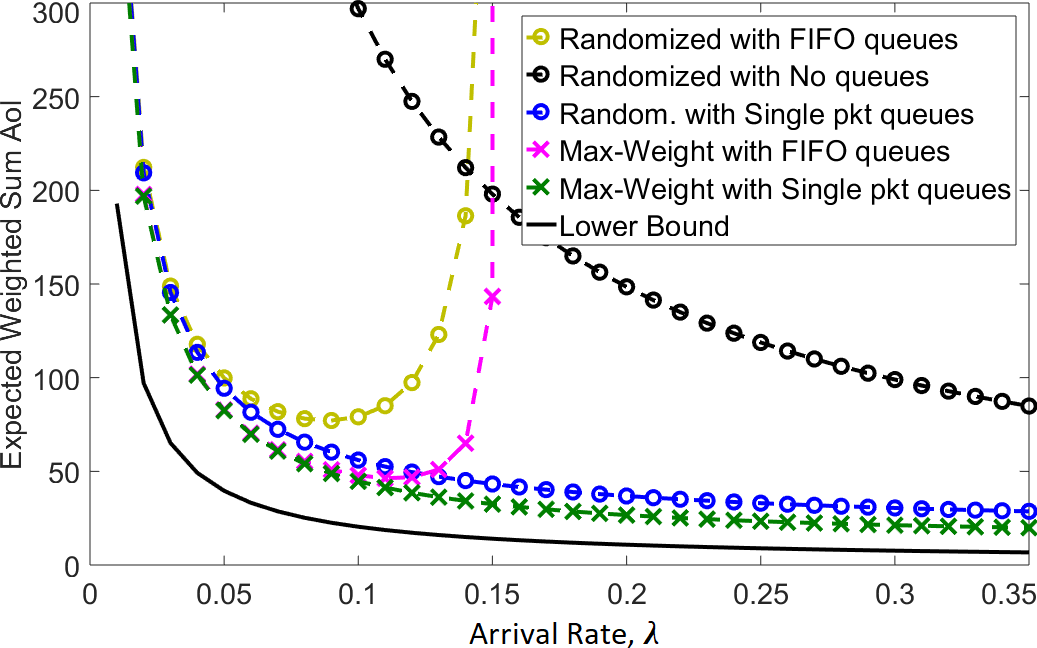}
\end{center}
\vspace{-0.2cm}
\caption{Simulation of networks with an increasing $\lambda$.} \label{fig.Var_A}
\vspace{-0.2cm}
\end{figure}

\begin{figure}[ht!]
\begin{center}
\includegraphics[height=5.3cm]{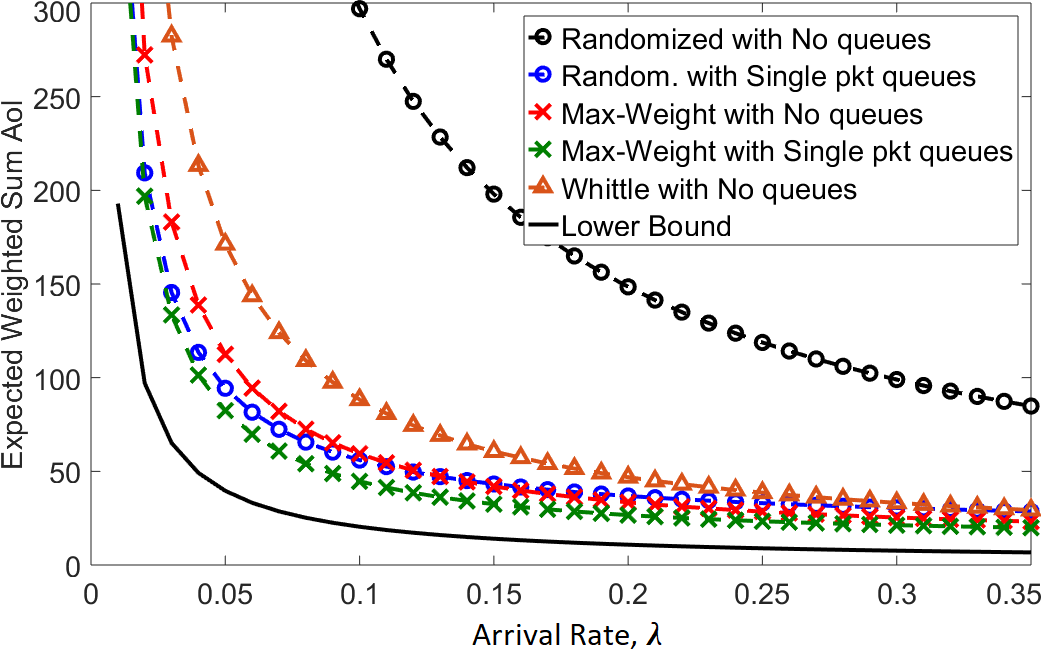}
\end{center}
\vspace{-0.2cm}
\caption{Simulation of networks with an increasing $\lambda$.} \label{fig.Var_A_1}
\vspace{-0.2cm}
\end{figure} 

The results in Figs.~\ref{fig.Var_A} and \ref{fig.Var_A_1} suggest that the Max-Weight policy outperforms the corresponding Randomized and Whittle's Index policies with the same queueing discipline for every value of $\lambda$. The results also show that under the same class of scheduling policies, \emph{Single packet queues} outperforms other queueing disciplines for every value of $\lambda$, as expected. It is evident from Fig.~\ref{fig.Var_A} that network instability, which occurs when $\lambda > 12/77$, is a major disadvantage of employing \emph{FIFO queues}.
\section{CONCLUDING REMARKS}\label{sec.Conclusion}
This paper considers a wireless network with a base station serving multiple traffic streams to different destinations. Packets from each stream arrive to the base station according to a Bernoulli process and are enqueued in separate (per stream) queues that could be of three types, namely \emph{FIFO queue}, \emph{Single packet queue} or \emph{No queue}, depending on the queueing discipline. \textcolor{black}{Notice that, from the perspective of AoI, Single packet queues are equivalent to LIFO queues.} We studied the problem of optimizing scheduling decisions with respect to the Expected Weighted Sum AoI of the network. Our main contributions include i) deriving a lower bound on the AoI performance achievable by any given network operating under any queueing discipline; ii) developing both an Optimal Stationary Randomized policy and a Max-Weight policy under each queueing discipline; and iii) evaluating the combined impact of the stochastic arrivals, queueing discipline and scheduling policy on the AoI using analytical and numerical results. 
We show that, contrary to intuition, the Optimal Stationary Randomized policy for \emph{Single packet queues} is insensitive to packet arrival rates. Simulation results show that the performance of the Age-Based Max-Weight policy for \emph{Single packet queues} is close to the analytical lower bound. Interesting extensions of this work include consideration of multi-hop networks and channels with unknown or time-varying statistics.

\section{ACKNOWLEDGMENT}
This work was supported by NSF Grants AST-1547331, CNS-1713725, and CNS-1701964, and by Army Research Office (ARO) grant number W911NF-17-1-0508.

\bibliographystyle{ACM-Reference-Format}
\bibliography{references}

\newpage
\appendix
\onecolumn
\section{Proof of Proposition \ref{prop.alternative_EWSAoI}}\label{app.alternative_EWSAoI}
\noindent \textbf{Proposition \ref{prop.alternative_EWSAoI}}. The infinite-horizon AoI objective function $J^\pi$ can be expressed as follows
\begin{equation}
J^\pi=\lim_{T\rightarrow\infty}\sum_{i=1}^N\frac{w_i}{2N}\left[\frac{\mathbb{\bar{M}}[I_i^2]}{\mathbb{\bar{M}}[I_i]}+\frac{2\mathbb{\bar{M}}[z_iI_i]}{\mathbb{\bar{M}}[I_i]}+1\right] \mbox{ w.p.1} \; ,
\end{equation}
where $I_i[m]$ is the inter-delivery time, $z_i[m]$ is the packet delay and 
\begin{equation}
\mathbb{\bar{M}}[z_iI_i]=\frac{1}{D_i(T)}\sum_{m=1}^{D_i(T)}z_i[m-1]I_i[m] \; . 
\end{equation}


\begin{proof}
Consider a network employing policy $\pi\in\Pi$ during the finite time-horizon $T$. Let $\Omega$ be the sample space associated with this network and let $\omega\in\Omega$ be a sample path. For a given sample path $\omega$, let $D_i(T)$ be the total number of packets delivered to destination $i$, $z_i[m]$ be the packet delay associated with the $m$th packet delivery, $I_i[m]$ be the number of slots between the $(m-1)$th and $m$th packet deliveries and $R_i$ be the number of slots remaining after the last packet delivery. Then, the time-horizon can be written as follows
\begin{equation}\label{eq.horizon_division}
T=\sum_{m=1}^{D_i(T)} I_i[m] + R_i, \forall i \in \{1,2,\cdots,N\} \; .
\end{equation}

The evolution of $h_i(t)$ is well-defined in each of the time intervals $I_i[m]$ and $R_i$. According to \eqref{eq.destination_AoI_1}, during the interval $I_i[m]$, the parameter $h_i(t)$ evolves as $\{z_i[m-1]+1,z_i[m-1]+2,\cdots,z_i[m-1]+I_i[m]\}$. This pattern is repeated throughout the entire time-horizon, for $m\in\{1,2,\cdots,D_i(T)\}$, and also during the last $R_i$ slots. As a result, the time-average AoI associated with destination $i$ can be expressed as
\begin{align}\label{eq.sum}
\frac{1}{T}\sum_{t=1}^{T}h_i(t)&=\frac{1}{T}\left[\sum_{m=1}^{D_i(T)}z_i[m-1]I_i[m]+\sum_{m=1}^{D_i(T)} \frac{(I_i[m]+1)I_i[m]}{2} +z_i[D_i(T)]R_i+ \frac{(R_i+1)R_i}{2}\right] \nonumber\\
&=\frac{1}{2}\left[\frac{D_i(T)}{T}\;\frac{1}{D_i(T)}\sum_{m=1}^{D_i(T)} \left(I_i^2[m]+2z_i[m-1]I_i[m]\right)+\frac{R_i^2}{T}+2\frac{z_i[D_i(T)]R_i}{T}+1\right] , \forall i \; ,
\end{align}
where the second equality uses \eqref{eq.horizon_division} to replace the two linear terms by $T$.

Combining \eqref{eq.horizon_division} with the sample mean $\bar{\mathbb{M}}[I_i]$, yields
\begin{equation}\label{eq.relation_throughput}
\frac{T}{D_i(T)}=\frac{\sum_{j=1}^{D_i(T)} I_i[j] + R_i}{D_i(T)}=\bar{\mathbb{M}}[I_i] + \frac{R_i}{D_i(T)} \; .
\end{equation}
Substituting \eqref{eq.relation_throughput} into \eqref{eq.sum} and then employing the sample mean operator $\mathbb{\bar{M}}$ on $I_i^2[m]$ and $z_i[m-1]I_i[m]$, gives
\begin{equation}\label{eq.sum1}
\frac{1}{T}\sum_{t=1}^{T}h_i(t)=\frac{1}{2}\left[\left(\bar{\mathbb{M}}[I_i] + \frac{R_i}{D_i(T)}\right)^{-1}\left(\bar{\mathbb{M}}[I_i^2]+2\bar{\mathbb{M}}[z_iI_i]\right)+\frac{R_i^2}{T}+2\frac{z_i[D_i(T)]R_i}{T}+1\right], \forall i \; ,
\end{equation}

The next step is to take the limit of \eqref{eq.sum1} as $T\rightarrow\infty$. Prior to taking the limit, we assume in the remaining part of this proof that the system time of the HoL packet in queue $i$ is finite, $z_i(t)<\infty$, as $t\rightarrow\infty$, with probability one. Recall from the discussion in Sec.~\ref{sec.QueueStability} that if $z_i(t)\rightarrow\infty$ with a positive probability, then the objective function diverges, $\mathbb{E}[J^\pi]\rightarrow\infty$. Hence, there is no loss of optimality in assuming that $z_i(t)<\infty$ with probability one. From this assumption, it follows that packet delays are finite with probability one, $z_i[m]<\infty$, and that packets are continuously delivered to destination $i$, what makes the number of slots after the last packet delivery $R_i$, finite with probability one. Hence, in the limit as $T\rightarrow\infty$, we have continuous packet deliveries, $D_i(T)\rightarrow\infty$, and finite $z_i[m]$ and $R_i$ implying that $R_i^2/T\rightarrow 0$, $R_i/D_i(T)\rightarrow 0$ and $z_i[D_i(T)]R_i/T\rightarrow 0$. Employing those limits into \eqref{eq.sum1} gives
\begin{equation}\label{eq.alternative_single}
\lim_{T\rightarrow\infty}\frac{1}{T}\sum_{t=1}^{T}h_i(t)=\lim_{T\rightarrow\infty}\left[\frac{\mathbb{\bar{M}}[I_i^2]}{2\mathbb{\bar{M}}[I_i]}+\frac{\mathbb{\bar{M}}[z_iI_i]}{\mathbb{\bar{M}}[I_i]}+\frac{1}{2}\right], \forall i \; .
\end{equation}
To obtain the final expression in \eqref{eq.alternative_EWSAoI} we employ \eqref{eq.alternative_single} into \eqref{eq.Objective_1}, without the expectation.
\end{proof}
\newpage
\section{Proof of Theorem \ref{theo.LowerBound}}\label{app.LowerBound}
\noindent \textbf{Theorem \ref{theo.LowerBound}} (Lower Bound). For any given network with parameters $(N,p_i,\lambda_i,w_i)$ and an arbitrary queueing discipline, the optimization problem in \eqref{eq.LowerBound_1}-\eqref{eq.LowerBound_3} provides a lower bound on the AoI minimization problem, namely $L_B\leq\mathbb{E}[J^*]$. The unique solution to \eqref{eq.LowerBound_1}-\eqref{eq.LowerBound_3} is given by
\begin{equation}
\hat{q}_i^{L_B}=\min{\left\{\lambda_i,\sqrt{\frac{w_ip_i}{2N\gamma^*}}\right\}}, \forall i \; ,
\end{equation}
where $\gamma^*$ yields from Algorithm~\ref{alg.LowerBound}. The lower bound is given by 
\begin{equation}
L_B=\frac{1}{2N}\sum_{i=1}^Nw_i\left(\frac{1}{\hat{q}_i^{L_B}}+1\right) \; .
\end{equation}

\begin{algorithm}
\renewcommand{\thealgorithm}{1}
\caption{Solution to the Lower Bound}
\begin{algorithmic}[1]
\State $\tilde{\gamma} \gets (\sum_{i=1}^N\sqrt{w_i/p_i})^2/(2N)$ and $\gamma_i \gets w_ip_i/2N\lambda_i^2 , \forall i$
\State $\gamma \gets \max\{\tilde{\gamma};\gamma_i\}$
\State $q_i \gets \lambda_i\min\{1 ; \sqrt{\gamma_i/\gamma}\} , \forall i$
\State $S \gets \sum_{i=1}^Nq_i/p_i$
\While{$S<1$ \textbf{and} $\gamma>0$}
\State decrease $\gamma$ slightly
\State repeat steps 4 and 5 to update $q_i$ and $S$
\EndWhile
\State \textbf{return} $\gamma^*=\gamma$ and $\hat{q}_i^{L_B}=q_i,\forall i$
\end{algorithmic}
\end{algorithm}

\begin{proof}
Consider a network with parameters $(N,p_i,\lambda_i,w_i)$ and an arbitrary queueing discipline. First, we show that \eqref{eq.LowerBound_1}-\eqref{eq.LowerBound_3} provides a lower bound $L_B$ on the AoI minimization problem $\mathbb{E}[J^*]=\min_{\pi\in\Pi}\mathbb{E}\left[J^\pi\right]$, then we find the unique solution to \eqref{eq.LowerBound_1}-\eqref{eq.LowerBound_3} by analyzing its KKT Conditions. The optimization problem in \eqref{eq.LowerBound_1}-\eqref{eq.LowerBound_3} is rewritten below for convenience.
\begin{subequations}
\begin{align*}
L_B=&\min_{\pi\in\Pi}\left\{\frac{1}{2N}\sum_{i=1}^Nw_i\left(\frac{1}{\hat{q}_i^{\pi}}+1\right)\right\} \\
\mbox{s.t. } &\textstyle\sum_{i=1}^N \hat{q}_i^{\pi}/p_i \leq 1 \; ;\\
&\hat{q}_i^{\pi} \leq \lambda_i , \forall i \; .
\end{align*}
\end{subequations}

Consider the expression for the time-average AoI associated with destination $i$ in \eqref{eq.sum}, which is valid for any admissible policy $\pi\in\Pi$ and time-horizon $T$. Substituting the non-negative terms $z_i[m-1]I_i[m]$ and $z_i[D_i(T)]R_i$ by zero, employing the sample mean operator $\mathbb{\bar{M}}$ to $I_i^2[m]$ and then applying Jensen's inequality $\mathbb{\bar{M}}[I_i^2]\geq(\mathbb{\bar{M}}[I_i])^2$, we obtain
\begin{equation}\label{eq.sum2}
\frac{1}{T}\sum_{t=1}^{T}h_i(t) \geq  \frac{1}{2}\left(\frac{D_i(T)}{T}\left(\bar{\mathbb{M}}[I_i]\right)^2+\frac{R_i^2}{T}+1\right) \; .
\end{equation}
Substituting \eqref{eq.relation_throughput} into \eqref{eq.sum2}, gives
\begin{equation}\label{eq.sum3}
\frac{1}{T}\sum_{t=1}^{T}h_i(t) \geq  \frac{1}{2}\left(\frac{1}{T}\frac{(T-R_i)^2}{D_i(T)}+\frac{R_i^2}{T}+1\right) \; .
\end{equation}
By minimizing the LHS of \eqref{eq.sum3} analytically with respect to the variable $R_i$, we have
\begin{equation}\label{eq.sum4}
\frac{1}{T}\sum_{t=1}^{T}h_i(t) \geq  \frac{1}{2}\left(\frac{T}{D_i(T)+1}+1\right) \; .
\end{equation}
Taking the expectation of \eqref{eq.sum4} and applying Jensen's inequality, yields
\begin{equation}\label{eq.sum5}
\frac{1}{T}\sum_{t=1}^{T}\mathbb{E}\left[h_i(t)\right] \geq  \frac{1}{2}\left(\frac{1}{\mathbb{E}\left[\displaystyle\frac{D_i(T)}{T}\right]+\displaystyle\frac{1}{T}}+1\right) \; .
\end{equation}
Applying the limit $T\rightarrow\infty$ to \eqref{eq.sum5} and using the definition of throughput in \eqref{eq.throughput}, gives
\begin{equation}\label{eq.sum6}
\lim_{T\rightarrow\infty}\frac{1}{T}\sum_{t=1}^{T}\mathbb{E}\left[h_i(t)\right] \geq  \frac{1}{2}\left(\frac{1}{\hat{q}_i^\pi}+1\right) \; .
\end{equation}
Substituting \eqref{eq.sum6} into the objective function in \eqref{eq.Objective_1}, yields 
\begin{align}\label{eq.obj1}
\mathbb{E}\left[J^{\pi}\right]=&\lim_{T\rightarrow\infty}\frac{1}{N}\sum_{i=1}^N\frac{w_i}{T}\sum_{t=1}^{T}\mathbb{E}\left[ h_i(t) \right]\nonumber \\
\geq&\frac{1}{2N}\sum_{i=1}^N w_i \left(\frac{1}{\hat{q}_i^\pi}+1\right)\; .
\end{align}
Inequality \eqref{eq.obj1} is valid for any admissible policy $\pi\in\Pi$. Notice that the RHS of \eqref{eq.obj1} depends only on the network's long-term throughput $\{\hat{q}_i^{\pi}\}_{i=1}^N$. Adding to \eqref{eq.obj1} the two necessary conditions for the long-term throughput in \eqref{eq.throughput_condition1} and \eqref{eq.throughput_condition2}, and then minimizing the resulting problem over all policies in $\Pi$, yields $\mathbb{E}[J^*]=\min_{\pi\in\Pi}\mathbb{E}\left[J^\pi\right]\geq L_B$ where $L_B$ is given by \eqref{eq.LowerBound_1}-\eqref{eq.LowerBound_3}.

After showing that \eqref{eq.LowerBound_1}-\eqref{eq.LowerBound_3} provides a lower bound on the AoI minimization problem, we find the unique set of network's long-term throughput $\{\hat{q}_i^{L_B}\}_{i=1}^N$ that solves \eqref{eq.LowerBound_1}-\eqref{eq.LowerBound_3} by analyzing its KKT Conditions. Let $\gamma$ be the KKT multiplier associated with the relaxation of $\sum_{i=1}^N \hat{q}_i^{\pi}/p_i \leq 1$ and $\{\zeta_i\}_{i=1}^{N}$ be the KKT multipliers associated with the relaxation of $\hat{q}_i^{\pi} \leq \lambda_i , \forall i$. Then, for $\gamma\geq 0$ , $\zeta_i\geq 0$ and $\hat{q}_i^\pi\in(0,1], \forall i$, we define
\begin{align}
\mathcal{L}(\hat{q}_i^\pi,&\zeta_i,\gamma)=\frac{1}{2N}\sum_{i=1}^Nw_i\left(\frac{1}{\hat{q}_i^\pi}+1\right)+\sum_{i=1}^N\zeta_i\left(\hat{q}_i^\pi-\lambda_i\right)+\gamma\left(\sum_{i=1}^N\frac{\hat{q}_i^{\pi}}{p_i}-1\right) \; ,
\end{align}
and, otherwise, we define $\mathcal{L}(\hat{q}_i^\pi,\zeta_i,\gamma)=+\infty$. Then, the KKT Conditions are
\begin{enumerate}[(i)]
\item Stationarity: $\nabla_{\hat{q}_i^\pi}\mathcal{L}(\hat{q}_i^\pi,\zeta_i,\gamma)=0$;
\item Complementary Slackness: $\gamma(\sum_{i=1}^N\hat{q}_i^{\pi}/p_i-1)=0$;
\item Complementary Slackness: $\zeta_i(\hat{q}_i^\pi-\lambda_i)=0, \forall i$;
\item Primal Feasibility: $\hat{q}_i^{\pi} \leq \lambda_i , \forall i$, and $\sum_{i=1}^N \hat{q}_i^{\pi}/p_i \leq 1$;
\item Dual Feasibility: $\zeta_i\geq 0, \forall i$, and $\gamma\geq 0$.
\end{enumerate}
Since $\mathcal{L}(\hat{q}_i^\pi,\zeta_i,\gamma)$ is a \emph{convex function}, if there exists a vector $(\{\hat{q}_i^{L_B}\}_{i=1}^N,\{\zeta_i^*\}_{i=1}^N,\gamma^*)$ that satisfies all KKT Conditions, then this vector is unique. 
Next, we find the vector $(\{\hat{q}_i^{L_B}\}_{i=1}^N,\{\zeta_i^*\}_{i=1}^N,\gamma^*)$.

To assess stationarity, $\nabla_{\hat{q}_i^\pi}\mathcal{L}(\hat{q}_i^\pi,\zeta_i,\gamma)=0$, we calculate the partial derivative of $\mathcal{L}(\hat{q}_i^\pi,\zeta_i,\gamma)$ with respect to $\hat{q}_i^\pi$, which gives
\begin{equation}\label{eq.KKT_equation1}
-\frac{w_ip_i}{2N(\hat{q}_i^\pi)^2}+\zeta_ip_i+\gamma=0 \; , \forall i \; .
\end{equation}

From complementary slackness, $\gamma(\sum_{i=1}^N\hat{q}_i^{\pi}/p_i-1)=0$, we know that either $\gamma=0$ or $\sum_{i=1}^N\hat{q}_i^{\pi}/p_i=1$. First, we consider the case $\sum_{i=1}^N\hat{q}_i^{\pi}/p_i=1$. Based on dual feasibility, $\zeta_i\geq 0$, we can separate streams $i \in \{1,\cdots,N\}$ into two categories: streams with $\zeta_i>0$ and streams with $\zeta_i=0$.\\
\noindent\emph{Category 1)} stream $i$ with $\zeta_i > 0$. It follows from complementary slackness, $\zeta_i(\hat{q}_i^\pi-\lambda_i)=0$, that $\hat{q}_i^\pi= \lambda_i$. Plugging this value of $\hat{q}_i^\pi$ into \eqref{eq.KKT_equation1} gives the inequality $\zeta_ip_i=\gamma_i-\gamma> 0$, where we define the constant
\begin{equation}\label{eq.KKT_equation3}
\gamma_i:=\frac{w_ip_i}{2N\lambda_i^2} \; .
\end{equation}
\noindent\emph{Category 2)} stream $i$ with $\zeta_i=0$. It follows from \eqref{eq.KKT_equation1} that
\begin{equation}\label{eq.KKT_solution2}
\gamma=\gamma_i\left(\frac{\lambda_i}{\hat{q}_i^\pi}\right)^2 \Rightarrow \; \hat{q}_i^\pi= \lambda_i\sqrt{\frac{\gamma_i}{\gamma}}\; , \mbox{for } \gamma_i-\gamma\leq 0 \; .
\end{equation}
Hence, for any fixed value of $\gamma\geq 0$, if $\gamma\geq\gamma_i$ then stream $i$ is in Category 2, otherwise, stream $i$ is in Category 1. Moreover, the values of $\zeta_i$ and $\hat{q}_i^\pi$ associated with stream $i$, in either Category, can be expressed as
\begin{align}
\zeta_i=\max\left\{0 ; \frac{\gamma_i-\gamma}{p_i}\right\} ,\forall i \; . \label{eq.lambda} \\
\hat{q}_i^\pi=\lambda_i\min\left\{1 ; \sqrt{\frac{\gamma_i}{\gamma}}\right\} ,\forall i \; . \label{eq.mu}
\end{align}
Notice that when $\gamma>\max\{\gamma_i\}$, then all streams are in Category 2 and $\hat{q}_i^\pi< \lambda_i,\forall i$. By decreasing the value of $\gamma$ gradually, the throughput $\hat{q}_i^\pi$ of each stream $i$ in \eqref{eq.mu} either increases or remain fixed at $\lambda_i$. Our goal is to find the value of $\gamma^*$ which yields $\{\hat{q}_i^{\pi}\}_{i=1}^N$ satisfying the condition $\sum_{i=1}^N\hat{q}_i^{\pi}/p_i=1$. Suppose this condition is satisfied when $\gamma>\max\{\gamma_i\}$, with all streams in Category 2, then it follows that
\begin{align}
&\sum_{i=1}^N\frac{\hat{q}_i^\pi}{p_i}=\sum_{i=1}^N\frac{\lambda_i}{p_i}\sqrt{\frac{\gamma_i}{\gamma}}=\frac{1}{\sqrt{2N\gamma}}\sum_{i=1}^N\sqrt{\frac{w_i}{p_i}}=1 \Rightarrow\nonumber\\
& \Rightarrow \gamma^*=\tilde{\gamma}:=\frac{1}{2N}\left(\sum_{i=1}^N\sqrt{\frac{w_i}{p_i}}\right)^2 \; ,
\label{eq.condition1}
\end{align}
where $\tilde{\gamma}$ is a fixed constant and the solution is unique $\gamma^*=\tilde{\gamma}$.

Alternatively, suppose that $\sum_{i=1}^N\hat{q}_i^{\pi}/p_i=1$ is satisfied when $\min\{\gamma_i\}\leq\gamma\leq\max\{\gamma_i\}$, with some streams in Category 1 and others in Category 2. To find $\gamma^*$, we start with $\gamma=\max\{\gamma_i\}$ and gradually decrease $\gamma$, adjusting $\{\hat{q}_i^\pi\}_{i=1}^N$ according to \eqref{eq.mu} until we reach $\sum_{i=1}^N\hat{q}_i^{\pi}/p_i=1$. The uniqueness of $\gamma^*$ follows from the monotonicity of $\hat{q}_i^\pi$ with respect to $\gamma$ in \eqref{eq.mu}.

Another possibility is for $\gamma$ to reach a value lower than $\min\{\gamma_i\}$ and still result in $\sum_{i=1}^N\hat{q}_i^{\pi}/p_i<1$. Notice from \eqref{eq.mu} that when $\gamma<\min\{\gamma_i\}$, then all streams are in Category 1 and have maximum throughputs, namely $\hat{q}_i^{\pi}=\lambda_i,\forall i$. It follows that $\sum_{i=1}^N\hat{q}_i^{\pi}/p_i=\sum_{i=1}^N\lambda_i/p_i<1$, in which case the condition $\sum_{i=1}^N\hat{q}_i^{\pi}/p_i=1$ cannot be satisfied for any value of $\gamma\geq 0$. Hence, from complementary slackness, $\gamma(\sum_{i=1}^N\hat{q}_i^{\pi}/p_i-1)=0$, we have the unique solution $\gamma^*=0$.

\emph{Proposed algorithm to find }$\gamma^*$ that solves the KKT Conditions: start with $\gamma=\max\{\gamma_i;\tilde{\gamma}\}$. Then, compute $\{\hat{q}_i^\pi\}_{i=1}^N$ using \eqref{eq.mu} and verify if the condition $\sum_{i=1}^N\hat{q}_i^{\pi}/p_i=1$ is satisfied. If $\sum_{i=1}^N\hat{q}_i^{\pi}/p_i<1$, then gradually decrease $\gamma$ and repeat the procedure. Stop when $\sum_{i=1}^N\hat{q}_i^{\pi}/p_i=1$ or when $\gamma<\min\{\gamma_i\}$. If $\sum_{i=1}^N\hat{q}_i^{\pi}/p_i=1$ holds, then assign $\gamma^*\leftarrow\gamma$. Otherwise, if $\gamma<\min\{\gamma_i\}$ holds, then assign $\gamma^*\leftarrow 0$. The solution to the KKT Conditions is given by $\gamma^*$ and the associated $\zeta_i^*$ and $\hat{q}_i^{L_B}$ obtained by substituting $\gamma^*$ into \eqref{eq.lambda} and \eqref{eq.mu}, respectively.

It is evident from the proposed algorithm that for any given network with parameters $(N,p_i,\lambda_i,w_i)$ and an arbitrary queueing discipline, the solution to the KKT Conditions, $(\{\hat{q}_i^{L_B}\}_{i=1}^N,\{\zeta_i^*\}_{i=1}^N,\gamma^*)$, exists and is unique. The proposed algorithm is described using pseudocode in Algorithm~\ref{alg.LowerBound}.

\end{proof}
\newpage
\section{Proof of Proposition \ref{prop.EWSAoI_randomized_OptimalBuffer}}\label{app.EWSAoI_randomized_OptimalBuffer}
\noindent \textbf{Proposition \ref{prop.EWSAoI_randomized_OptimalBuffer}}. The optimal EWSAoI achieved by a network with Single packet queues over the class $\Pi_R$ is given by
\begin{tcolorbox}[title=Optimal Randomized policy for Single packet queues,left=1mm,right=1mm,top=-2mm,bottom=0mm]
\begin{subequations}
\begin{align}
\mathbb{E}\left[J^{R^S}\right]=&\min_{R\in\Pi_R}\left\{\frac{1}{N} \sum_{i=1}^N w_i \left(\frac{1}{\lambda_i}-1+\frac{1}{p_i\mu_i}\right)\right\} \\
\mbox{s.t. } &\textstyle\sum_{i=1}^N \mu_i \leq 1 \; ;
\end{align}
\end{subequations}
\end{tcolorbox}
\noindent where $R^S$ denotes the Optimal Stationary Randomized Policy for the Single packet queue discipline.

\begin{proof}
Consider the evolution of $h_i(t)$ and $z_i^S(t)$ given in \eqref{eq.destination_AoI_1} and \eqref{eq.system_time_Optimal}, respectively. Under policy $R\in\Pi_R$, the tuple $(h_i(t),z_i^S(t))$ 
evolves according to a two-dimensional Markov Chain with countably-infinite state space that fully characterizes the state of stream $i$. The basic structure of this Markov Chain is illustrated in Fig.~\ref{fig.Two_Dim_MC}.

\begin{figure}[ht!]
\begin{center}
\includegraphics[height=4cm]{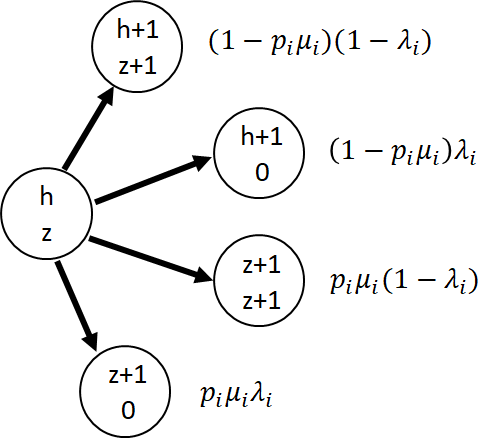}
\end{center}
\caption{Illustration of the state evolution associated with stream $i$ of a network employing policy $R\in\Pi_R$ and operating under the Single packet queue discipline. In partciular, we show the outgoing transition arcs from any given state $(h_i(t),z_i(t))=(h,z),\forall z\in\{0,1,2,\cdots\},h\geq z$, with the associated transition probabilities.} \label{fig.Two_Dim_MC}
\end{figure} 

To obtain the stationary distribution of stream $i$'s Markov Chain, we separate its state transitions into three categories and derive the associated probability distributions.
\begin{itemize}
\item Transition to state $(h,h),\forall h\in\{1,2,\cdots\}$, following a successful packet transmission and no packet arrival (i.e. transition to an empty\footnote{When the queue is empty, the system time of the HoL packet $z$ is not part of the network state. However, to facilitate the analysis, we assume in this proof that $z$ is always part of the network state and evolves according to \eqref{eq.system_time_Optimal}.} queue):
\begin{equation}\label{eq.trans_to_empty}
\mathbb{P}(h,h)=\mathbb{P}(1,0)\frac{(1-\lambda_i)^h}{\lambda_i}\left\{\frac{1-(1-p_i\mu_i)^h}{p_i\mu_i}\right\} \; ;
\end{equation}
\item Transition to state $(h,0),\forall h\in\{1,2,\cdots\}$, following a packet arrival:
\begin{equation}\label{eq.trans_to_arrival}
\mathbb{P}(h,0)=\mathbb{P}(1,0)\left\{\sum_{n=0}^{h-1}(1-\lambda_i)^{h-1-n}(1-p_i\mu_i)^{n}\right\} \; ;
\end{equation}
\item Uneventful transition to state $(h,z),\forall z\in\{1,2,\cdots\},h>z$:
\begin{align}\label{eq.uneventful_trans}
\mathbb{P}(h,z)&=\mathbb{P}(h-z,0)(1-\lambda_i)^z(1-p_i\mu_i)^z\nonumber\\
				&=\mathbb{P}(1,0)(1-\lambda_i)^z(1-p_i\mu_i)^z\left\{\sum_{n=0}^{h-z-1}(1-\lambda_i)^{h-z-1-n}(1-p_i\mu_i)^{n}\right\}\; .
\end{align}
\end{itemize}
Notice that \eqref{eq.trans_to_empty}, \eqref{eq.trans_to_arrival} and \eqref{eq.uneventful_trans} comprehend all possible state transitions. With the probability distributions, we obtain an expression for the probability of the event $h_i(t)=h$
\begin{align}
\mathbb{P}(h)=\sum_{z=0}^{h}\mathbb{P}(h,z)=\frac{\mathbb{P}(1,0)}{\lambda_i}\left[\sum_{n=0}^{h-1}(1-\lambda_i)^{h-1-n}(1-p_i\mu_i)^n\right] \; ,
\end{align} 
for $h\geq 1$. Moreover, since $\sum_h\mathbb{P}(h)=1$, we have that $\mathbb{P}(1,0)=\lambda_i^2p_i\mu_i$.

The countable-state Markov Chain is irreducible and has a stationary distribution, hence this distribution is unique, the chain is positive recurrent and
\begin{equation}\label{eq.stationary_h}
\lim_{T\rightarrow\infty}\frac{1}{T}\sum_{t=1}^T\mathbb{E}[h_i(t)]=\mathbb{E}[h]=\sum_{h} h\mathbb{P}(h)=\frac{1}{p_i\mu_i}+\frac{1}{\lambda_i}-1 \; .
\end{equation}
Proposition~\ref{prop.EWSAoI_randomized_OptimalBuffer} follows from substituting \eqref{eq.stationary_h} into the objective function in \eqref{eq.AoI_optimal}.
\end{proof}

\newpage
\section{Proof of Theorem \ref{theo.MW_bound_Optimal}}\label{app.MW_bound_Optimal}
\noindent \textbf{Theorem \ref{theo.MW_bound_Optimal}} (Performance Bounds for $MW^S$). Consider a network employing Single packet queues. The performance of the Max-Weight policy with $\beta_i=w_i/p_i\mu_i^S,\forall i$, is such that
\begin{equation}
\mathbb{E}\left[J^{MW^S}\right]\leq\mathbb{E}\left[J^{R^S}\right]\; ,
\end{equation}
where $\mu_i^S$ and $\mathbb{E}[J^{R^S}]$ are the optimal scheduling probability for the case of Single packet queues and the associated EWSAoI attained by $R^S$, respectively.

\begin{proof}
Consider stream $i$ from a network operating under the \emph{Single packet queue} discipline. In each slot $t$, a packet is transmitted, i.e. $u_i(t)=1$, if the stream is selected and its queue is non-empty. Hence, packet transmissions $u_i(t)$ depend on the queue backlog. To decouple packet transmissions from the queue backlog, we create \emph{dummy packets} that can be transmitted without affecting the AoI. In particular, suppose that at time $t$ queue $i$ is selected and successfully transmits a packet with $z_i^S(t)=z$. Then, at the beginning of slot $t+1$, with probability $1-\lambda_i$ we place a \emph{dummy packet} with $z_i^S(t+1)=z+1$ at the HoL of the queue, otherwise we place a real packet with $z_i^S(t)=0$. 
From that moment on, the behavior of dummy packets is indistinguishable from real packets. Notice that due to the choice of $z_i^S(t+1)=z+1$, when a dummy packet is delivered to the destination, it does not change the associated AoI. Moreover, the system time $z_i^S(t)$ is now defined at every slot $t$ following \eqref{eq.system_time_Optimal}. Next, we analyze the equivalent system with dummy packets.

The Age-Based Max-Weight policy minimizes the drift in \eqref{eq.LyapunovDrift_1}. Hence, any other policy $\pi \in \Pi$ yields a higher (or equal) value of $\Delta(\mathbb{S}(t))$. Consider the Stationary Randomized policy for Single packet queues defined in Sec.~\ref{sec.Random_OptimalBuffer} with scheduling probability $\mu_i^S$ and let 
\begin{equation}\label{eq.uis}
\mathbb{E}\left[u_i(t)|\mathbb{S}(t)\right]=\mathbb{E}\left[u_i\right]=\mu_i^S \; .
\end{equation}
Substituting $\mu_i^S$ into the Lyapunov Drift gives the upper bound
\begin{align}\label{eq.drift_first}
\Delta(\mathbb{S}(t))\leq\frac{1}{N}\sum_{i=1}^N\beta_i-\frac{1}{N}\sum_{i=1}^N\beta_ip_i\left(h_i(t)-z_i^S(t)\right)\mu_i^S \; . 
\end{align}

Now, taking the expectation with respect to $\mathbb{S}(t)$ and then the time-average on the interval $t\in\{1,2,\cdots,T\}$ yields
\begin{align}
\frac{\mathbb{E}\left[L(T+1)\right]}{T}-\frac{\mathbb{E}\left[L(1)\right]}{T}\leq \frac{1}{N}\sum_{i=1}^N\beta_i-\frac{1}{TN}\sum_{t=1}^T\sum_{i=1}^N\beta_ip_i\mathbb{E}\left[h_i(t)-z_i^S(t)\right]\mu_i^S \; .
\end{align}
Manipulating this expression, assigning $\beta_i=w_i/p_i\mu_i^S$ and taking the limit as $T\rightarrow\infty$, gives
\begin{align}\label{eq.MW_perf}
\mathbb{E}\left[J^{MW^S}\right]\leq&\frac{1}{N}\sum_{i=1}^N\frac{w_i}{p_i\mu_i^S}+\lim_{T\rightarrow\infty}\frac{1}{TN}\sum_{i=1}^N\sum_{t=1}^Tw_i\mathbb{E}\left[z_i^S(t)\right] \; .
\end{align}
From the evolution of $z_i^S(t)$ in \eqref{eq.system_time_Optimal}, we know that
\begin{equation}\label{eq.z_Singlequeue}
\lim_{T\rightarrow\infty}\frac{1}{TN}\sum_{i=1}^N\sum_{t=1}^Tw_i\mathbb{E}\left[z_i^S(t)\right]=\frac{1}{N}\sum_{i=1}^Nw_i\left(\frac{1}{\lambda_i}-1\right) \; .
\end{equation}
Substituting \eqref{eq.z_Singlequeue} into \eqref{eq.MW_perf} and then comparing the result with \eqref{eq.EWSAoI_randomized_OptimalBuffer} yields
\begin{align}\label{eq.MW_perf_2}
\mathbb{E}\left[J^{MW^S}\right]\leq\frac{1}{N}\sum_{i=1}^N\frac{w_i}{p_i\mu_i^S}+\frac{1}{N}\sum_{i=1}^Nw_i\left(\frac{1}{\lambda_i}-1\right)=\mathbb{E}\left[J^{R^S}\right] \; .
\end{align}

\end{proof}

\newpage
\section{Proof of Theorem \ref{theo.MW_bound_NoBuffer}}\label{app.MW_bound_NoBuffer}
\noindent \textbf{Theorem \ref{theo.MW_bound_NoBuffer}} (Performance Bounds for $MW^N$). Consider a network employing the No queues discipline. The performance of the Max-Weight Policy with $\beta_i=w_i/p_i\mu_i^N,\forall i$, is such that
\begin{equation}
\mathbb{E}\left[J^{MW^N}\right]\leq\mathbb{E}\left[J^{R^N}\right]  \; ,
\end{equation}
where $\mu_i^N$ and $\mathbb{E}[J^{R^N}]$ are the optimal scheduling probability for the case of No queues and the associated EWSAoI attained by $R^N$, respectively.

\begin{proof}
Consider stream $i$ from a network operating under the \emph{No queue} discipline. In each slot $t$, a packet is successfully transmitted, i.e. $d_i(t)=1$, if a packet arrives, the stream is selected and the channel is ON. Notice that all delivered packets have $z_i^N(t)=0$. This is equivalent to a network with packets that are always fresh, i.e. $z_i^N(t)=0,\forall i,t$, and with a virtual channel that is ON with probability $p_i\lambda_i$ and OFF with probability $1-p_i\lambda_i$. 
The Lyapunov Drift for this equivalent system with fresh packets and virtual channels is given by:
\begin{align}\label{eq.LyapunovDrift_AppE}
\Delta(\mathbb{S}(t))=\frac{1}{N}\sum_{i=1}^N\hat\beta_i-\frac{1}{N}\sum_{i=1}^N\hat\beta_i\lambda_ip_ih_i(t)\mathbb{E}\left[ \left. u_i(t) \right| \mathbb{S}(t) \right]\; . 
\end{align}
For minimizing $\Delta(\mathbb{S}(t))$, the \emph{Max-Weight policy selects, in each slot $t$, the stream $i$ with a HoL packet and the highest value of} $\hat\beta_i\lambda_ip_ih_i(t)$, with ties being broken arbitrarily. By comparing the drift of the equivalent system \eqref{eq.LyapunovDrift_AppE} and the original system \eqref{eq.LyapunovDrift_1}, it is easy to see that $\beta_i=\hat\beta_i\lambda_i$.

The Age-Based Max-Weight policy minimizes the drift in \eqref{eq.LyapunovDrift_AppE}. Hence, any other policy $\pi \in \Pi$ yields a higher (or equal) value of $\Delta(\mathbb{S}(t))$. Consider the Stationary Randomized policy for No queues defined in Sec.~\ref{sec.Random_NoBuffer} with scheduling probability $\mu_i^N$ and let 
\begin{equation}\label{eq.uiN}
\mathbb{E}\left[u_i(t)|\mathbb{S}(t)\right]=\mathbb{E}\left[u_i\right]=\mu_i^N \; .
\end{equation}
Substituting $\mu_i^N$ into the Lyapunov Drift gives the upper bound
\begin{align}\label{eq.drift_first_N}
\Delta(\mathbb{S}(t))\leq\frac{1}{N}\sum_{i=1}^N\hat\beta_i-\frac{1}{N}\sum_{i=1}^N\hat\beta_i\lambda_ip_ih_i(t)\mu_i^N \; . 
\end{align}

Now, taking the expectation with respect to $\mathbb{S}(t)$ and then the time-average on the interval $t\in\{1,2,\cdots,T\}$ yields
\begin{align}
\frac{\mathbb{E}\left[L(T+1)\right]}{T}-\frac{\mathbb{E}\left[L(1)\right]}{T}\leq \frac{1}{N}\sum_{i=1}^N\hat\beta_i-\frac{1}{TN}\sum_{t=1}^T\sum_{i=1}^N\hat\beta_i\lambda_ip_i\mathbb{E}\left[h_i(t) \right]\mu_i^N \; .
\end{align}
Manipulating this expression, assigning $\hat\beta_i=w_i/\lambda_ip_i\mu_i^N$ and taking the limit as $T\rightarrow\infty$, gives
\begin{align}\label{eq.MW_perf_N}
\mathbb{E}\left[J^{MW^N}\right]\leq&\frac{1}{N}\sum_{i=1}^N\frac{w_i}{\lambda_ip_i\mu_i^N} \; .
\end{align}

For deriving the upper bound in \eqref{eq.MW_bound_NoBuffer}, consider the Optimal Stationary Randomized policy $R^N$. Substituting $\mu_i^N$ into \eqref{eq.EWSAoI_randomized_NoBuffer} and then comparing with \eqref{eq.MW_perf_N} gives
\begin{align}\label{eq.MW_perf_3_N}
\mathbb{E}\left[J^{MW^N}\right]\leq&\mathbb{E}\left[J^{R^N}\right]\; .
\end{align}
\end{proof}

\end{document}